\documentclass[12pt]{article}
\usepackage{amsthm}
\usepackage{amsmath}
\usepackage{amssymb}
\usepackage{latexsym}
\usepackage{enumitem}
\usepackage{comment}
\usepackage{xcolor}
\usepackage{authblk}

\newtheorem{definition}{Definition}[section]
\newtheorem{theorem}{Theorem}[section]

\newtheorem{proposition}{Proposition}[section]

\newtheorem{example}{Example}[section]
\newtheorem{corollary}{Corollary}[section]

\newtheorem{remark}{Remark}[section]

\usepackage[backend=biber,style=apa,citestyle=authoryear,uniquename=false,uniquelist=false]{biblatex}
\usepackage[colorlinks]{hyperref}

\title{Information divergences of Markov chains and their applications}
\author[1]{Youjia Wang\thanks{Email: e1124868@u.nus.edu}}
\author[2]{Michael C.H. Choi\thanks{Email: mchchoi@nus.edu.sg}}
\affil[1]{Department of Statistics and Data Science, National University of Singapore, Singapore}
\affil[2]{Department of Statistics and Data Science and Yale-NUS College, National University of Singapore, Singapore}
\date{\today}

\begin{document}

\maketitle

\begin{abstract}
    In this paper, we first introduce and define several new information divergences in the space of transition matrices of finite Markov chains which measure the discrepancy between two Markov chains. These divergences offer natural generalizations of classical information-theoretic divergences, such as the $f$-divergences and the R\'enyi divergence between probability measures, to the context of finite Markov chains. We begin by detailing and deriving fundamental properties of these divergences and notably gives a Markov chain version of the Pinsker's inequality and Chernoff information. We then utilize these notions in a few applications. First, we investigate the binary hypothesis testing problem of Markov chains, where the newly defined R\'enyi divergence between Markov chains and its geometric interpretation play an important role in the analysis. Second, we propose and analyze information-theoretic (Ces\`aro) mixing times and ergodicity coefficients, along with spectral bounds of these notions in the reversible setting. Examples of the random walk on the hypercube, as well as the connections between the critical height of the low-temperature Metropolis-Hastings chain and these proposed ergodicity coefficients, are highlighted. 
\end{abstract}

\section{Introduction}

Given two Markov chains with transition matrices $M,L$ on a common finite state space $\mathcal{X}$, how to measure their information-theoretic discrepancies? The present manuscript is largely motivated by this question and one of the main aims is to systematically analyze various natural notions of information divergences of Markov chains and to present some interesting and important applications, particularly in the context of mixing times, convergence towards equilibrium and hypothesis testing of Markov chains. 

Let us begin by briefly recalling several established notions of information divergences of Markov chains in the literature. In \parencite{billera2001geometric}, the total variation distance between $M,L$ with respect to a given probability distribution $\pi$ is introduced and investigated for Metropolis-Hastings chains, namely 
$$\mathrm{TV}(M,L) :=\dfrac{1}{2}\sum_{x,y\in\mathcal{X}}\pi(x)\left|M(x,y)-L(x,y)\right|,$$
while the Kullback-Leibler (KL) divergence from $L$ to $M$ with respect to $\pi$, as studied in \cite{wolfer2021information}, is given by
$$D_{KL}(M \| L) := \sum_{x,y\in\mathcal{X}}\pi(x)M(x,y) \ln \left(\dfrac{M(x,y)}{L(x,y)}\right).$$
These two divergences are in fact special cases of $f$-divergences. Let $f: \mathbb{R}^+ \to \mathbb{R}$ be a convex function with $f(1) = 0$ that we call it as the generator of $f$-divergence. The $f$-divergence from $L$ to $M$ with respect to $\pi$ is defined to be
$$D_f(M\|L):=\sum_{x\in \mathcal{X}}\pi(x)\sum_{y\in\mathcal{X}}L(x,y)f\left(\dfrac{M(x,y)}{L(x,y)}\right).$$
Note that by taking $f(t) = \frac{1}{2}|t-1|$ and $f(t) = t \ln t$, we recover respectively the total variation distance and KL divergence introduced earlier. 

An important family of $f$-divergences is known as the $\alpha$-divergence, that we denote by $D_{\alpha}$, for $\alpha \in (0,1) \cup (1,+\infty)$ by taking $f(t)=\frac{t^\alpha-1}{\alpha-1}$. The R\'enyi divergence from $L$ to $M$ of the order $\alpha\in (0,1)\cup (1,+\infty)$ can be defined to be
\begin{align*}
    R_{\alpha}(M\|L) &:=\dfrac{1}{\alpha-1}\ln \left(1+(\alpha-1) D_{\alpha}( M\| L)\right).
\end{align*}
All these divergences shall be introduced in a more precise manner in Section \ref{sec:basicdef}.

With the above notations in mind, we proceed to highlight several important and interesting achievements of this manuscript. The first major achievements lie in \textbf{introducing and analyzing fundamental properties of $D_f, D_{\alpha}$ and $R_{\alpha}$} in the context of discrete-time Markov chains. While these divergences are frequently utilized to quantify the information divergences of probability measures, to the best of our knowledge these divergences have not been systematically investigated in finite Markov chains. Notable highlights include
\begin{itemize}
    \item Markov chain Pinsker's inequality (Proposition \ref{inequalities with other divergences}): for $0<\alpha\leq 1$,
    \begin{equation*}
        \dfrac{\alpha}{2}\mathrm{TV}^2(M,L)\leq R_{\alpha}(M\|L).
    \end{equation*}

    \item Chernoff information of Markov chain (Proposition \ref{Chernoff information: informal}): for $\alpha > 0$, 
    \begin{equation*}
        (1-\alpha)R_{\alpha}(M\|L)=\inf_{P\in \mathcal{L}}\left\{\alpha D_{KL}(P\|M)+(1-\alpha)D_{KL}(P\|L)\right\},
    \end{equation*}
    where $\mathcal{L} = \mathcal{L}(\mathcal{X})$ is the set of all transition matrices on $\mathcal{X}$.
    
    \item Ergodicity coefficients in the space of finite transition matrices, and their properties (Proposition \ref{prop:ergocoef}):
    given $P \in \mathcal{L}$ and $\alpha\in (0,1)\cup (1,+\infty)$, we define
    \begin{align*}
        \eta^f(P) = \eta^f(P,\pi) := \sup_{M,L \in \mathcal{L};~ M \neq L}\dfrac{D_f(MP\|LP)}{D_f(M\|L)}, \\
        \eta^{(R_{\alpha})}(P) = \eta^{(R_{\alpha})}(P,\pi) := \sup_{M,L \in \mathcal{L};~ M \neq L}\dfrac{R_{\alpha}(MP\|LP)}{R_{\alpha}(M\|L)}.
    \end{align*}
\end{itemize}

In the remaining of this section, let us assume that $P$ is ergodic and admits $\pi$ as its stationary distribution. Let $\Pi$ be a matrix with each row being $\pi$ for all rows. To quantify and assess the convergence rate of $P^t$ towards $\Pi$, we define the $\pi$-weighted $\alpha$-divergence mixing time to be
\begin{align*}
    \bar{t}_{\mathrm{mix}}^{(D_\alpha)}(\varepsilon,P)&:=\inf\left\{t>0:D_{\alpha}(P^t\|\Pi)<\varepsilon\right\}.
\end{align*}
Analogously we shall define R\'enyi divergence mixing times. In this vein, one of our major achievements in this direction lie in \textbf{proposing and analyzing new information-theoretic mixing times}, in which We offer spectral bounds on these mixing times:

\begin{theorem}[Informal presentation of Theorem \ref{upper and lower bound of D_alpha}]\label{thm:informal1}
    Let $P$ be an irreducible, aperiodic and $\pi$-reversible transition matrix, and we denote $\pi_{\mathrm{min}}:=\min_{x}\pi(x)>0$. For $\alpha\in (0,1)\cup (1,+\infty)$ and sufficiently small $\varepsilon$, we have
    \begin{align*}
        \bar{t}_{\mathrm{mix}}^{(D_\alpha)}(\varepsilon,P)&\leq \begin{cases}
            \dfrac{1}{\gamma_*}\left(\ln \dfrac{2\alpha}{(\alpha-1)\pi_{\mathrm{min}}}+\ln \dfrac{1}{\varepsilon}\right), \quad \textrm{if }\, \alpha>1,\\
            \dfrac{1}{\gamma_*}\left(\ln \dfrac{8\alpha}{(1-\alpha)\pi_{\mathrm{min}}}+\ln \dfrac{1}{\varepsilon}\right), \quad \textrm{if }\, 0<\alpha<1,
        \end{cases}    
    \end{align*}
    and 
    \begin{align*}
        \bar{t}_{\mathrm{mix}}^{(D_\alpha)}(\varepsilon,P) &\geq \begin{cases}
            \dfrac{1}{2\ln \frac{1}{1-\gamma_*}}\left(\ln \dfrac{\pi_{\mathrm{min}}^2}{8}+\ln\dfrac{1}{\varepsilon}\right), \quad \textrm{if }\, \alpha>1,\\
            \dfrac{1}{2\ln \frac{1}{1-\gamma_*}}\left(\ln \dfrac{\pi_{\mathrm{min}}^2\alpha}{16}+\ln \dfrac{1}{\varepsilon}\right), \quad \textrm{if }\, 0<\alpha<1,
        \end{cases}     
    \end{align*}
    where $\gamma_*$ is the absolute spectral gap of $P$. 
\end{theorem}

On the other hand, the worst-case mixing time based upon $\alpha$-divergence is defined to be, for $\varepsilon > 0$,
\begin{align*}
     t_{\mathrm{mix}}^{(\mathrm{D_{\alpha}})}(\varepsilon,P)&:=\inf\left\{t>0: \max_{x\in \mathcal{X}}\widetilde D_{\alpha}(P^t(x,\cdot)\|\pi)=\max_{x\in \mathcal{X}}\dfrac{1}{\alpha-1}\left(\sum_{y\in \mathcal{X}}\pi(y)\left(\dfrac{P^t(x,y)}{\pi(y)}\right)^{\alpha}-1\right) <\varepsilon\right\}.
\end{align*}
To compare the average mixing time $\bar{t}_{\mathrm{mix}}^{(D_\alpha)}(\varepsilon,P)$ and the worst-case mixing time $t_{\mathrm{mix}}^{(\mathrm{D_{\alpha}})}(\varepsilon,P)$, in the following result we prove that these two mixing times are of the same order under very mild assumptions on the absolute spectral gap of $P$ (which are for instance satisfied in low-temperature Metropolis-Hastings chain or the random walk on hypercube):
\begin{theorem}[Informal presentation of Theorem \ref{comparison between TV and classical}]
    Under the same assumptions and notations as in Theorem \ref{thm:informal1}, and in addition $\gamma_*(P)$ satisfies some mild assumptions, we have
    \begin{equation*}
            \dfrac{1}{5}t_{\mathrm{mix}}^{(D_\alpha)}(\varepsilon,P)
            \leq \bar{t}_{\mathrm{mix}}^{(D_\alpha)}(\varepsilon,P)\leq t_{\mathrm{mix}}^{(D_\alpha)}(\varepsilon,P).
    \end{equation*}
\end{theorem}

Analogous results comparing the average and the worst-case mixing time based upon R\'enyi divergence and total variation distance are also provided in Theorem \ref{comparison between TV and classical}. We also remark on the series of work that compare the order of various mixing times and hitting times of Markov chains or Markov processes, see for example \parencite{aldous1982some, anderson2023mixing, hermon2015technical, basu2014characterization, griffiths2014tight, oliveira2012mixing, peres2015mixing, sousi2014mixing} and the references therein.

The rest of this paper is organized as follows. In Section \ref{sec:basicdef}, we formally introduce various information divergences and ergodicity coefficients of Markov chains. Their fundamental properties are stated and proved in Section \ref{sec:properties}, followed by applications in binary hypothesis testing of Markov chains in Section \ref{sec:hypothesis testing}. The information-theoretic mixing times and their spectral bounds are discussed in Section \ref{mixing time to stationary distribution}. Finally, we elucidate on the relationship between Ces\`aro mixing times, ergodicity coefficients and the critical height of the low-temperature Metropolis-Hastings chain in Section \ref{cesaro mixing time}.

\section{Basic definitions}\label{sec:basicdef}

On a finite state space $\mathcal{X}$, we define $\mathcal{L} = \mathcal{L}(\mathcal{X})$ as the set of transition matrices of discrete-time Markov chains. We denote by $\mathcal{P}(\mathcal{X})$ to be the set of probability masses on $\mathcal{X}$. Let $\pi \in \mathcal{P}$ is any given positive probability distribution on $\mathcal{X}$, and denote $\mathcal{L}(\pi)\subset \mathcal{L}$ as the set of $\pi$-reversible transition matrices, where a transition matrix $P \in \mathcal{L}$ is said to be $\pi$-reversible if $\pi(x)P(x,y)=\pi(y)P(y,x)$ for all $x,y\in \mathcal{X}$. We also say that $P \in \mathcal{L}$ is $\pi$-stationary if it satisfies $\pi P = \pi$. 

First, we give the definition of $f$-divergence of Markov chains and recall that of probability measures.

\begin{definition}[$f$-divergence of Markov chains and of probability measures]\label{def:f-divergence}

Let $f:\mathbb{R}^+\mapsto \mathbb{R}$ be a convex function and $f(1)=0$. 
For given $\pi \in \mathcal{P}(\mathcal{X})$ and transition matrices $M, L\in \mathcal{L}$, we define the $f$-divergence from $L$ to $M$ as 
\begin{equation}\label{f divergence for matrices}
     D_f(M\|L):=\sum_{x\in \mathcal{X}}\pi(x)\sum_{y\in\mathcal{X}}L(x,y)f\left(\dfrac{M(x,y)}{L(x,y)}\right).
\end{equation}
For two probability measures $\mu,\nu \in \mathcal{P}(\Omega)$ with $\Omega$ finite, the $f$-divergence from $\nu$ to $\mu$ is defined to be 
\begin{equation}\label{f divergence for probability measure}
    \widetilde D_f(\mu\|\nu):= \sum_{x \in \Omega} \nu(x)f\left(\frac{\mu(x)}{\nu(x)}\right),
\end{equation}
where we apply the usual convention that $0f(\frac{0}{0}):=0$ and $0f(\frac{a}{0}):= a f^{\prime}(+\infty)$ with $f^{\prime}(+\infty) := \lim_{x \to 0^+} x f(\frac{1}{x})$ for $a > 0$ in the two definitions above.
\end{definition}

Note that $D_f(M\|L)$ can be understood as a $\pi$-weighted average of $\widetilde D_f(M(x,\cdot)\|L(x,\cdot))$, since we have
\begin{align}\label{eq:weightedDfML}
    D_f(M\|L) = \sum_{x \in \mathcal{X}} \pi(x) \widetilde D_f(M(x,\cdot)\|L(x,\cdot)).
\end{align}

In the sequel, a majority of our focus is devoted to the special case of $\alpha$-divergences, where we take $f(t)=\frac{t^\alpha-1}{\alpha-1}$ with $\alpha\in (0,1)\cup (1,+\infty)$ in Definition \ref{def:f-divergence}. As a result, this yields
\begin{align}
    D_{\alpha}(M\|L)&:=\dfrac{1}{\alpha-1}\left(\sum_{x\in \mathcal{X}}\pi(x)\sum_{y\in\mathcal{X}}L(x,y)\left(\dfrac{M(x,y)}{L(x,y)}\right)^{\alpha}-1\right),\\
    \widetilde D_{\alpha}(\mu\|\nu)&:=\dfrac{1}{\alpha-1}\left(\sum_{x\in \Omega}\nu(x)\left(\dfrac{\mu(x)}{\nu(x)}\right)^\alpha-1\right),\label{alpha divergence between discrete probability measures}
\end{align}
where the same convention as in Definition \ref{def:f-divergence} still applies in the equations above. Using $\widetilde D_{\alpha}(\mu\|\nu)$, we now introduce the R\'enyi divergence between probability measures as systematically investigated in \parencite{van2014renyi}.

\begin{definition}[R\'enyi divergence of probability measures]\label{Renyi divergence for probability measures}

For probability measures $\mu,\nu \in \mathcal{P}(\Omega)$ and $\Omega$ is finite, the R\'enyi divergence from $\nu$ to $\mu$ of the order $\alpha\in (0,1)\cup (1,+\infty)$ is given by
\begin{align}
    \widetilde R_{\alpha}(\mu\|\nu)&:= \dfrac{1}{\alpha-1}\ln \sum_{x\in \Omega}\nu(x)\left(\dfrac{\mu(x)}{\nu(x)}\right)^\alpha \nonumber \\
    &= \dfrac{1}{\alpha-1}\ln \left(1+(\alpha-1)\widetilde D_{\alpha}(\mu\|\nu)\right).\label{Renyi and f divergence for probability measure} 
\end{align}
\end{definition}

Using simple convexity arguments, it can readily be shown that $\widetilde R_{\alpha}(\mu\|\nu) \geq 0$ (see for example \parencite{van2014renyi}), from which we see that $1+(\alpha-1)\widetilde D_{\alpha}(\mu\|\nu)\geq 0$. In the case for Markov chains however, the definition for R\'enyi divergence from $L$ to $M$ is more involved, as we cannot directly define it in the form \eqref{Renyi and f divergence for probability measure} via simply replacing $\widetilde D_{\alpha}(\mu\|\nu)$ with $D_{\alpha}(M\|L)$, since such definition can be ill-defined as it is a priori unclear whether $1+(\alpha-1)D_{\alpha}(M\|L)\geq0$ when $\alpha<1$. To give a meaningful definition in this case, one possible remedy is to instead consider the divergence of the edge measure of the two chains.
\begin{definition}[Edge measure of a Markov chain with respect to $\pi$]\label{transformation from matrix to probability measure} 

For $x,y\in \mathcal{X}$, $M\in\mathcal{L}$ and $\pi \in \mathcal{P}(\mathcal{X})$, if we write
\begin{align*}
    \widetilde M(x,y)&:=\pi(x)M(x,y),
\end{align*}
then $\widetilde M = \left(\widetilde M(x,y)\right)_{\mathcal{X} \times \mathcal{X}}$ is a probability measure on $\mathcal{X}\times \mathcal{X}$. $\widetilde M$ is known as the edge measure of $M$ with respect to $\pi$, see \parencite{wolfer2021information}. Analogously we can define the edge measure of $L$ with respect to $\pi$, denoted by $\widetilde L$.
\end{definition}

Using Definition \ref{transformation from matrix to probability measure} and \eqref{alpha divergence between discrete probability measures}, this immediately yields an important equality which will play a key role in the following sections:
\begin{equation}\label{equality}
    \widetilde D_{\alpha}(\widetilde M\|\widetilde L)=\dfrac{1}{\alpha-1}\left(\sum_{x,y\in \mathcal{X}}\pi(x)L(x,y)\left(\dfrac{M(x,y)}{L(x,y)}\right)^\alpha-1\right)=D_{\alpha}(M\|L).
\end{equation}
Intuitively, it means that the $\alpha$-divergence from the edge measure $\widetilde L$ to $\widetilde M$ equals to the $\alpha$-divergence from the chain $L$ to $M$. Thus, to analyze the property of $D_{\alpha}(M\|L)$, it suffices to investigate the corresponding property of $\widetilde D_{\alpha}(\widetilde M\|\widetilde L)$. With equality \eqref{equality}, we can proceed to give a well-defined R\'enyi divergence $R_{\alpha}(M\|L)$. 

\begin{definition}[R\'enyi divergence of Markov chains]

Given $M,L\in \mathcal{L}$, the R\'enyi divergence from $L$ to $M$ of the order $\alpha\in (0,1)\cup (1,+\infty)$ is given by
\begin{align}
    R_{\alpha}(M\|L)&:=\widetilde R_{\alpha}(\widetilde M\|\widetilde L)\label{equality for Renyi divergence}\\
    &=\dfrac{1}{\alpha-1}\ln \left(1+(\alpha-1)\widetilde D_{\alpha}(\widetilde M\|\widetilde L)\right)\nonumber\\
    &=\dfrac{1}{\alpha-1}\ln \left(1+(\alpha-1) D_{\alpha}( M\| L)\right).\label{eq:Renyi dovergence for transition matrices}
\end{align}
\end{definition}

Note that for $t \geq 0$ (resp.~$0 \leq t \leq 1/(1-\alpha)$) and $\alpha > 1$ (resp.~$\alpha \in (0,1)$), the mapping $t \mapsto \frac{1}{\alpha-1} \ln(1+(\alpha-1)t)$ is concave (resp.~convex), and hence by Jensen's inequality and \eqref{eq:weightedDfML}, this yields
\begin{align}
    R_{\alpha}(M\|L) \begin{cases}
            \geq \sum_{x \in \mathcal{X}} \pi(x) \widetilde R_{\alpha}(M(x,\cdot)\|L(x,\cdot)), \quad \textrm{if }\, \alpha>1,\\
            \leq \sum_{x \in \mathcal{X}} \pi(x) \widetilde R_{\alpha}(M(x,\cdot)\|L(x,\cdot)), \quad \textrm{if }\, 0<\alpha<1,
        \end{cases}
\end{align}
where the right hand side above can be understood as the $\pi$-weighted average of the R\'enyi divergence $\widetilde R_{\alpha}(M(x,\cdot)\|L(x,\cdot))$. These two inequalities also contrast with the equality in \eqref{eq:weightedDfML}.

In addition to $D_{\alpha}$ and $R_{\alpha}$ as discussed earlier, we next introduce some $f$-divergences of Markov chains. Their counterparts in probability measures are well-studied in the literature, see for example \parencite[Eq.5-7]{van2014renyi}.

\begin{definition}[Some common $f$-divergences of Markov chains]\label{other divergences}

For $M,L\in\mathcal{L}$ and $\pi \in \mathcal{P}(\mathcal{X})$, by taking $f(t)$ to be respectively $(\sqrt{t}-1)^2, (t-1)^2, |t-1|/2$ in Definition \ref{def:f-divergence}, we define 
\begin{align}
    (\mathbf{Hellinger\enspace distance})\quad \mathrm{Hel}^2(M,L)&:=\sum_{x,y\in\mathcal{X}}\pi(x)\left(\sqrt{M(x,y)}-\sqrt{L(x,y)}\right)^2,\\
    (\mathbf{\chi^2-divergence})\quad \chi^2(M\|L)&:=\sum_{L>0}\pi(x)\dfrac{\left(M(x,y)-L(x,y)\right)^2}{L(x,y)},\\
    (\mathbf{Total\enspace variation\enspace distance})\quad \mathrm{TV}(M,L)&:=\dfrac{1}{2}\sum_{x,y\in\mathcal{X}}\pi(x)\left|M(x,y)-L(x,y)\right|.
\end{align}
\end{definition}

The next definition concerns the closest $\pi$-reversible transition matrix to any given transition matrix $L$ as measured by the R\'{e}nyi divergence of the order $\alpha$:
\begin{definition}[$\alpha$-projection]\label{alpha-projection}

For $L\in \mathcal{L}$ and $\alpha\in (0,1)\cup (1,+\infty)$, we say that $L^{(R_\alpha)}$ is an $\alpha$-projection of $L$ to $\mathcal{L}(\pi)$, if
\begin{equation*}
   L^{(R_\alpha)}=L^{(R_\alpha)}(L,\pi)\in\mathop{\arg\min}\limits_{M\in\mathcal{L}(\pi)}R_{\alpha}(M\|L),
\end{equation*}
where we recall that $\mathcal{L}(\pi)$ denotes the set of $\pi$-reversible transition matrices.
\end{definition}

Note that the set $\mathcal{L}(\pi)$ is a compact subset of $\mathbb{R}^{|\mathcal{X}| \times |\mathcal{X}|}$ and the mapping $M \mapsto R_{\alpha}(M\|L)$ is clearly continuous, and hence $L^{(R_\alpha)}$ exists. We shall investigate the uniqueness of $\alpha$-projection later using the Pythagorean inequality \eqref{pythagorean inequality} in Proposition \ref{prop:pythagorean inequality and weak triangle inequaliy}.

The final definition of the section introduces new ergodicity coefficients of $P$ on the space of transition matrices $\mathcal{L}$ based upon information divergences. These can be used to assess the contraction rate of $P$.

\begin{definition}[Ergodicity coefficients of Markov chains based on $D_f$ and $R_{\alpha}$]\label{def:ergodicity coefficients}
    Given $P \in \mathcal{L}$, $\pi \in \mathcal{P}(\mathcal{X})$ and $\alpha\in (0,1)\cup (1,+\infty)$, we define
    \begin{align*}
        \eta^f(P) = \eta^f(P,\pi) := \sup_{M,L \in \mathcal{L};~ M \neq L}\dfrac{D_f(MP\|LP)}{D_f(M\|L)}, \\
        \eta^{(R_{\alpha})}(P) = \eta^{(R_{\alpha})}(P,\pi) := \sup_{M,L \in \mathcal{L};~ M \neq L}\dfrac{R_{\alpha}(MP\|LP)}{R_{\alpha}(M\|L)}.
    \end{align*}
    Their classical counterparts on the space of probability measures $\mathcal{P}(\mathcal{X})$ are defined analogously to be
    \begin{align*}
        \widetilde\eta^f(P) := \sup_{\mu,\nu \in \mathcal{P}(\mathcal{X});~ \mu \neq \nu}\dfrac{\widetilde D_f(\mu P\|\nu P)}{\widetilde D_f(\mu\|\nu)}, \\
        \widetilde\eta^{(R_{\alpha})}(P) := \sup_{\mu,\nu \in \mathcal{P}(\mathcal{X});~ \mu \neq \nu} \dfrac{\widetilde R_{\alpha}(\mu P\|\nu P)}{\widetilde R_{\alpha}(\mu\|\nu)}.
    \end{align*}
\end{definition}

In the special cases of total variation distance and $\alpha$-divergence, we shall write the corresponding ergodicity coefficients of $P$ to be $\eta^{(\mathrm{TV})}(P),\widetilde\eta^{(\mathrm{TV})}(P)$ for $f(t)=|t-1|/2$ and $\eta^{(D_{\alpha})}(P),\widetilde\eta^{(D_{\alpha})}(P)$ for $f(t)=\frac{t^\alpha-1}{\alpha-1}$ respectively. In the literature, $\eta^{(\mathrm{TV})}(P)$ is also called the Dobrushin contraction coefficient, see \parencite{cohen1993relative}.



\section{Properties of $D_{\alpha}(M\|L)$ and $R_{\alpha}(M\|L)$}\label{sec:properties}
In this section, we present some fundamental and important properties of the $\alpha$-divergence $D_{\alpha}(M\|L)$ and the R\'enyi divergence $R_{\alpha}(M\|L)$ for $M,L\in \mathcal{L}$, which follows readily from the analogous results \parencite{van2014renyi} in the context of probability measures. For $M, L \in \mathcal{L}$, we give a roadmap of the major results and new notions that we introduce in this section as follows. We shall see that many of these properties are inherited from their counterparts of $\widetilde D_{\alpha}, \widetilde R_{\alpha}$:
\begin{itemize}
    \item Non-negativity of $D_{\alpha}(M\|L)$ and $R_{\alpha}(M\|L)$: Proposition \ref{prop:non-negative}
    
    \item Convexity of the mappings $(M,L)\mapsto D_{\alpha}(M\|L)$ and $(M,L)\mapsto R_{\alpha}(M\|L)$: Proposition \ref{prop:convexity}

    \item Monotonicity in $\alpha$ of $R_{\alpha}(M\|L)$: Proposition \ref{prop:monotonicity}

    \item Pinsker's inequality of Markov chains and various relationships with other divergences: Proposition \ref{inequalities with other divergences}

    \item Pythagorean inequality and weak triangle inequality of the R\'{e}nyi divergence: Proposition \ref{prop:pythagorean inequality and weak triangle inequaliy}

    \item Chernoff information of Markov chains: Proposition \ref{Chernoff information: informal}

    \item Properties of various ergodicity coefficients: Proposition \ref{prop:ergocoef}
    
\end{itemize}

The first result that we present demonstrates the non-negativity of $D_{\alpha}(M\|L)$ and $R_{\alpha}(M\|L)$.

\begin{proposition}[Non-negativity of $D_{\alpha}(M\|L)$ and $R_{\alpha}(M\|L)$]\label{prop:non-negative}
    
Given $M,L\in \mathcal{L}$ and $\pi \in \mathcal{P}(\mathcal{X})$ is strictly positive on $\mathcal{X}$, we have $D_{\alpha}(M\|L)\geq 0$ and $R_{\alpha}(M\|L)\geq 0$, and the equality holds if and only if $M=L$.
\end{proposition}
\begin{proof}
    In view of \eqref{eq:weightedDfML}, the non-negativity of $D_{\alpha}$ is clear and inherits from that of $\widetilde D_{\alpha}$, see e.g. \cite[Theorem $7.5$]{polyanskiy2022information}. Since $\pi$ is strictly positive, the equality holds if and only if $\widetilde D_{\alpha}(M(x,\cdot)\|L(x,\cdot)) = 0$ for all $x \in \mathcal{X}$ if and only if $M=L$. 
    As to $R_{\alpha}(M\|L)$, since $R_{\alpha}(M\|L)=\widetilde R_{\alpha}\left(\widetilde M\|\widetilde L\right)$, it suffices to show $\widetilde R_{\alpha}\left(\widetilde M\|\widetilde L\right)\geq 0$, which is a classic result as in \parencite[Theorem 8]{van2014renyi}.
\end{proof}

The second result proves the convexity of the $\alpha$-divergence and R\'{e}nyi divergence for suitable $\alpha$ choices:
\begin{proposition}[Convexity of $D_{\alpha}(M\|L)$ and $R_{\alpha}(M\|L)$]\label{prop:convexity}

For $M,L\in \mathcal{L}$ and $\pi \in \mathcal{P}(\mathcal{X})$, the following statements hold:
\begin{enumerate}[label=(\roman*).]
    \item For $\alpha\in (0,1)$, both of $D_{\alpha}(M\|L)$ and $R_{\alpha}(M\|L)$ are jointly convex in $(M,L)$.

    \item For $\alpha\in (1,+\infty)$, $D_{\alpha}(M\|L)$ is jointly convex in $(M,L)$, and $R_{\alpha}(M\|L)$ is convex in $L$.

    \item For $\alpha\in (0,1)\cup (1,+\infty)$, $R_{\alpha}(M\|L)$ is jointly quasi-convex in (M,L), i.e. for $\lambda\in (0,1)$ and $(M_1,L_1),(M_2,L_2)$, we have
    \begin{equation*}
        R_{\alpha}(\lambda M_1+(1-\lambda)M_2\|\lambda L_1+(1-\lambda)L_2)\leq 
        \max \left\{R_{\alpha}(M_1\|L_1),R_{\alpha}(M_2\|L_2)\right\}.
    \end{equation*}
\end{enumerate}
\end{proposition}
\begin{proof}
    We first prove that for $\alpha>0$ and $\alpha\neq 1$, $D_{\alpha}(M\|L)$ is jointly convex in $(M,L)$, i.e. for $\lambda\in (0,1)$ and $(M_1,L_1),(M_2,L_2) \in \mathcal{L} \times \mathcal{L}$, we need to prove
    \begin{equation*}
        D_{\alpha}\left(\lambda M_1+(1-\lambda)M_2\|\lambda L_1+(1-\lambda)L_2\right)\leq 
        \lambda D_{\alpha}(M_1\|L_1)+(1-\lambda)D_{\alpha}(M_2\|L_2).
    \end{equation*}
    This can be readily verified: as $f(t)=\frac{t^{\alpha}-1}{\alpha-1}$ is convex, which implies that $L(x,y) f\left(\frac{M(x,y)}{L(x,y)}\right)$ is convex in $(M(x,y),L(x,y))$, since
    \begin{align*}
        L_{\lambda}(x,y)f\left(\dfrac{M_{\lambda}(x,y)}{L_{\lambda}(x,y)}\right) &\leq \lambda L_1(x,y) f\left(\dfrac{M_1(x,y)}{L_1(x,y)}\right)+(1-\lambda)L_2(x,y) f\left(\dfrac{M_2(x,y)}{L_2(x,y)}\right)\\
        \iff f\left(\dfrac{M_{\lambda}(x,y)}{L_{\lambda}(x,y)}\right) &\leq \lambda \dfrac{L_1(x,y)}{L_{\lambda}(x,y)}f\left(\dfrac{M_1(x,y)}{L_1(x,y)}\right)+(1-\lambda)\dfrac{L_2(x,y)}{L_{\lambda}(x,y)}f\left(\dfrac{M_2(x,y)}{L_2(x,y)}\right),
    \end{align*}
    where $M_{\lambda}(x,y)=\lambda M_1(x,y) + (1-\lambda)M_2(x,y)$ and $L_{\lambda}(x,y)=\lambda L_1(x,y) + (1-\lambda)L_2(x,y)$ are the shorthand notations, and the last inequality indeed holds from the convexity of $f$.

    When $\alpha \in (0,1)$, the mapping $t \mapsto \frac{1}{\alpha-1}\ln (1+(\alpha-1)t)$ is convex, hence $R_{\alpha}(M\|L)=\frac{1}{\alpha-1}\ln \left(1+(\alpha-1) D_{\alpha}( M\| L)\right)$ is also convex. 

    When $\alpha\in (1,+\infty)$, for all $\lambda\in (0,1)$ and $(M,L_1),(M,L_2) \in \mathcal{L} \times \mathcal{L}$, we need to prove
    \begin{equation*}
        R_{\alpha}(M\|\lambda L_1+(1-\lambda)L_2)\leq \lambda R_{\alpha}(M\|L_1)+(1-\lambda)R_{\alpha}(M\|L_2).
    \end{equation*}
     Note that for $L_{\lambda}=\lambda L_1+(1-\lambda)L_2$, we have
    \begin{align*}
        R_{\alpha}(M\|L_{\lambda})&= \widetilde R_{\alpha}\left(\widetilde M\|\widetilde L_{\lambda}\right)\\
        &=\widetilde R_{\alpha}\left(\widetilde M\|\lambda\widetilde L_1+(1-\lambda)\widetilde L_2\right)\\
        &\leq \lambda \widetilde R_{\alpha}\left(\widetilde M\|\widetilde L_1\right)+(1-\lambda)\widetilde R_{\alpha}\left(\widetilde M\|\widetilde L_2\right)\\
        &=\lambda R_{\alpha}(M\|L_1)+(1-\lambda)R_{\alpha}(M\|L_2),
    \end{align*}
    where the inequality follows from the convexity of $\widetilde R_{\alpha}(M\|L)$ in $L$ from \parencite[Theorem 12]{van2014renyi}, and hence $R_{\alpha}(M\|L)$ is convex in $L$. Similarly, the last item can also be proved via using $R_{\alpha}(M\|L)=\widetilde R_{\alpha}(\widetilde M\|\widetilde L)$.
\end{proof}

The next result proves the monotonicity of the R\'enyi divergence in $\alpha$:

\begin{proposition}[Monotonicity in $\alpha$]\label{prop:monotonicity}
    
    For $M,L\in \mathcal{L}$ and $\pi \in \mathcal{P}(\mathcal{X})$, $R_{\alpha}(M\|L)$ is non-decreasing with respect to $\alpha\in (0,1)\cup (1,+\infty)$. 
\end{proposition}
\begin{proof}
    For $\alpha,\beta\in (0,1)\cup (1,+\infty)$ and $\alpha<\beta$, the mapping $t \mapsto t^{\frac{\alpha-1}{\beta-1}}$ is strictly convex when $\alpha<1$ and strictly concave when $\alpha>1$, hence we have
    \begin{align*}
        R_{\alpha}(M\|L)&=\dfrac{1}{\alpha-1}\ln \left(\sum_{x,y\in \mathcal{X}}\pi(x)L(x,y)\left(\dfrac{M(x,y)}{L(x,y)}\right)^{\alpha}\right)\\
        &=\dfrac{1}{\alpha-1}\ln \left(\sum_{x,y\in \mathcal{X}}\pi(x)M(x,y)\left(\dfrac{M(x,y)}{L(x,y)}\right)^{\alpha-1}\right)\\
        &=\dfrac{1}{\alpha-1}\ln \left(\sum_{x,y\in \mathcal{X}}\pi(x)M(x,y)\left(\dfrac{M(x,y)}{L(x,y)}\right)^{(\beta-1)\frac{\alpha-1}{\beta-1}}\right)\\
        &\leq \dfrac{1}{\alpha-1}\ln \left(\sum_{x,y\in \mathcal{X}}\pi(x)M(x,y)\left(\dfrac{M(x,y)}{L(x,y)}\right)^{(\beta-1)}\right)^{\frac{\alpha-1}{\beta-1}}\\
        &=\dfrac{1}{\beta-1}\ln \left(\sum_{x,y\in \mathcal{X}}\pi(x)M(x,y)\left(\dfrac{M(x,y)}{L(x,y)}\right)^{(\beta-1)}\right)\\
        &=R_{\beta}(M\|L).
    \end{align*}
\end{proof}

We naturally extend the notion of R\'enyi divergence of Markov chains to order $0, 1, +\infty$ in the next result:

\begin{proposition}[Extension of $R_{\alpha}$ to $\alpha=0,1$ and $+\infty$]\label{extensions to alpha=0,1,infty}
    
    For $M,L\in\mathcal{L}$ and $\pi \in \mathcal{P}(\mathcal{X})$, we define 
    \begin{align*}
        R_0(M\|L):=\lim_{\alpha \searrow 0}R_{\alpha}(M\|L), \quad
        R_1(M\|L):=\lim_{\alpha \nearrow 1}R_{\alpha}(M\|L), \quad
        R_{\infty}(M\|L):=\lim_{\alpha \nearrow \infty}R_{\alpha}(M\|L),
    \end{align*}
    then we have 
    \begin{align*}
        R_0(M\|L)&=-\ln \left(\sum_{x,y;~M(x,y)>0}\pi(x)L(x,y)\right),\\
        R_1(M\|L)&=\sum_{x,y;~L(x,y)>0}\pi(x)M(x,y)\ln \dfrac{M(x,y)}{L(x,y)},\\
        R_{\infty}(M\|L)&=\ln \left(\max\limits_{x,y;~L(x,y)>0}\dfrac{M(x,y)}{L(x,y)}\right),
    \end{align*}
    where $R_1(M\|L)$ is the Kullback-Leibler (KL) divergence from $L$ to $M$, and we write $D_{KL}(M\|L):=R_{1}(M\|L)$.
\end{proposition}
\begin{proof}
    When $0<\alpha<1$, it is easy to see that $M(x,y)^{\alpha}L(x,y)^{1-\alpha}\leq \alpha M(x,y)+(1-\alpha)L(x,y)$, then by the dominated convergence theorem, we have 
    \begin{align*}
        \lim_{\alpha\searrow 0}R_{\alpha}(M\|L)&=\lim_{\alpha\searrow 0}\dfrac{1}{\alpha-1}\ln \sum_{x,y\in \mathcal{X}}\pi(x)L(x,y)\left(\dfrac{M(x,y)}{L(x,y)}\right)^{\alpha}\\
        &=-\ln \left(\lim_{\alpha\searrow 0}\sum_{M,L>0}\pi(x)L(x,y)\left(\dfrac{M(x,y)}{L(x,y)}\right)^{\alpha}\right)\\
        &=-\ln \left(\sum_{M,L>0}\pi(x)L(x,y)\right)\\
        &=-\ln \left(\sum_{M>0}\pi(x)L(x,y)\right).
    \end{align*}

Denote $x_{\alpha}=\sum_{x,y\in \mathcal{X}}\pi(x)L(x,y)\left(\frac{M(x,y)}{L(x,y)}\right)^{\alpha}$. Similar to the argument above, by the dominated convergence theorem we have $\lim_{\alpha\nearrow 1} x_{\alpha}=1$. Therefore, we have 
\begin{align*}
    \lim_{\alpha\nearrow 1}R_{\alpha}(M\|L)&=\lim_{\alpha\nearrow 1}\dfrac{\ln x_{\alpha}}{\alpha-1}\\
    &=\lim_{\alpha\nearrow 1}\dfrac{x_{\alpha}-1}{\alpha-1}\\
    &=\lim_{\alpha\nearrow 1}\sum_{L>0}\pi(x)M(x,y)\left(\dfrac{\left(\frac{M(x,y)}{L(x,y)}\right)^{\alpha-1}-1}{\alpha-1}\right).
\end{align*}
Given $t>0$, the mapping $g:\beta\mapsto \frac{t^{\beta}-1}{\beta}$ is non-decreasing for $\beta\in (-1,0)$, since $g^{\prime}(t)=\frac{t^{\beta}}{\beta^2}\left(\beta\ln t-(1-t^{-\beta})\right)=\frac{t^{\beta}}{\beta^2}\left((t^{-\beta}-1)-\ln t^{-\beta}\right)\geq 0$. Hence by monotone convergence theorem, we have
\begin{align*}
    \lim_{\alpha\nearrow 1}R_{\alpha}(M\|L)&=\lim_{\alpha\nearrow 1}\sum_{L>0}\pi(x)M(x,y)\left(\dfrac{\left(\frac{M(x,y)}{L(x,y)}\right)^{\alpha-1}-1}{\alpha-1}\right)\\
    &=\sum_{M,L>0}\pi(x)M(x,y)\lim_{\alpha\nearrow 1}\left(\dfrac{\left(\frac{M(x,y)}{L(x,y)}\right)^{\alpha-1}-1}{\alpha-1}\right)\\
    &=\sum_{M,L>0}\pi(x)M(x,y)\ln \dfrac{M(x,y)}{L(x,y)}.
\end{align*}
As to $R_{\infty}(M\|L)$, we have 
\begin{align*}
    \lim_{\alpha\nearrow \infty}R_{\alpha}(M\|L)&=\lim_{\alpha\nearrow\infty}\ln \left(\sum_{L>0}\pi(x)M(x,y)\left(\dfrac{M(x,y)}{L(x,y)}\right)^{\alpha-1}\right)^{\frac{1}{\alpha-1}}\\
    &=\ln \left(\max\limits_{L>0}\dfrac{M(x,y)}{L(x,y)}\right).
\end{align*}
\end{proof}

Building upon Proposition \ref{extensions to alpha=0,1,infty}, in the next result we prove various inequalities between information divergences of Markov chains, and in particular we give an analogue of Pinsker's inequality in this context.

\begin{proposition}[Inequalities with other divergences]\label{inequalities with other divergences}
    
For $M,L\in\mathcal{L}$, $\pi \in \mathcal{P}(\mathcal{X})$ and recall the $f$-divergences as in Definition \ref{other divergences}, we have
\begin{enumerate}[label=(\roman*).]
    \item $\mathrm{Hel}^2(M,L)\leq R_{\frac{1}{2}}(M\|L)\leq R_1(M\|L)\leq R_2(M\|L)\leq \chi^2(M\|L)$.

    \item (Pinsker's inequality) 
    
    For $0<\alpha\leq 1$, we have
    \begin{equation*}
        \dfrac{\alpha}{2}\mathrm{TV}^2(M,L)\leq R_{\alpha}(M\|L).
    \end{equation*}
\end{enumerate}
\end{proposition}
\begin{proof}
    First we observe that 
    \begin{align*}
        R_{\frac{1}{2}}(M\|L)&=-2\ln \left(1-\dfrac{\mathrm{Hel}^2(M,L)}{2}\right),\\
        R_2(M\|L)&=\ln \left(1+\chi^2(M\|L)\right),
    \end{align*}
    then by $\ln t\leq t-1$ for $t>0$, we have $\mathrm{Hel}^2(M,L)\leq R_{\frac{1}{2}}(M\|L)$ and $R_2(M\|L)\leq \chi^2(M\|L)$.

    For probability measures $\mu,\nu$ on finite state space $\Omega$, we denote $\widetilde {\mathrm{TV}}(\mu,\nu):=\dfrac{1}{2}\sum_{x\in \Omega}|\mu(x)-\nu(x)|$, then we have 
    \begin{align*}
        R_{\alpha}(M\|L) =\widetilde R_{\alpha}\left(\widetilde M\|\widetilde L\right)
        &\geq \dfrac{\alpha}{2}\widetilde{\mathrm{TV}}^2\left(\widetilde M\|\widetilde L\right)\\
        &=\dfrac{\alpha}{8}\left(\sum_{x,y\in\mathcal{X}}\left|\pi(x)M(x,y)-\pi(x)L(x,y)\right|\right)^2\\
        &=\dfrac{\alpha}{2}\mathrm{TV}^2(M\|L),
    \end{align*}
    where the inequality follows from \parencite[Theorem 31]{van2014renyi}.
\end{proof}

We proceed to prove the skew symmetry of equivalence of R\'enyi divergence of Markov chains:
\begin{proposition}[Skew symmetry and equivalence]\label{prop:skewsym}

Given $M,L\in \mathcal{L}$, $\pi \in \mathcal{P}(\mathcal{X})$ and $0<\alpha<\beta<1$, we have
\begin{enumerate}[label=(\roman*).]
    \item $R_{\alpha}(M\|L)=\dfrac{\alpha}{1-\alpha}R_{1-\alpha}(L\|M)$.

    \item $\dfrac{\alpha}{\beta}\dfrac{1-\beta}{1-\alpha}R_{\beta}(M\|L)\leq R_{\alpha}(M\|L)\leq R_{\beta}(M\|L)$.
\end{enumerate}
\end{proposition}
\begin{proof}
    Since 
    \begin{align*}
    (1-\alpha)R_{\alpha}(M\|L) &=-\ln \left(\sum_{M,L>0}\pi(x)M^{\alpha}(x,y)L^{1-\alpha}(x,y)\right),\\
    \alpha R_{1-\alpha}(L\|M) &=-\ln \left(\sum_{M,L>0}\pi(x)L^{1-\alpha}(x,y)M^{\alpha}(x,y)\right),
    \end{align*}
 the first item is verified. The second item follows from $\frac{1-\beta}{\beta}R_{\beta}(M\|L)=R_{1-\beta}(L\|M)$ and $\frac{1-\alpha}{\alpha}R_{\alpha}(M\|L)=R_{1-\alpha}(L\|M)$.
\end{proof}

We prove the Pythagorean inequality and weak triangle inequality of the R\'enyi divergence of Markov chains in the following Proposition:

\begin{proposition}[Pythagorean inequality and weak triangle inequality]\label{prop:pythagorean inequality and weak triangle inequaliy}
    
    Suppose we are given $L\in\mathcal{L}$ and recall that an $\alpha$-projection $L^{(R_\alpha)}\in\mathop{\arg\min}_{M\in\mathcal{L}(\pi)}R_{\alpha}(M\|L)$ is introduced in Definition \ref{alpha-projection} and exists. For any $M\in\mathcal{L}(\pi)$ and $\alpha>0$, we have the following Pythagorean inequality:
    \begin{equation}\label{pythagorean inequality}
        R_{\alpha}(M\|L)\geq R_{\alpha}(M\|L^{(R_\alpha)})+R_{\alpha}(L^{(R_\alpha)}\|L).
    \end{equation}
    In addition, for $\alpha>1$ and any $U,V,W\in \mathcal{L}$, we have the weak triangle inequality
    \begin{equation}\label{weak triangle inequality}
        R_{\alpha}(U\|V)\leq \dfrac{2\alpha-1}{2\alpha-2}R_{2\alpha}(U\|W)+R_{2\alpha-1}(W\|V).
    \end{equation}
\end{proposition}
\begin{proof}
    For $\pi \in \mathcal{P}(\mathcal{X})$, we denote the symmetric edge measure family on $\mathcal{X}\times\mathcal{X}$ as 
    \begin{equation*}
        \mathcal{E}(\pi):=\left\{\pi(x)g(x,y): g(x,y)\geq 0, \sum_{y}g(x,y)=1, \pi(x)g(x,y)=\pi(y)g(y,x)\right\}.
    \end{equation*}
    Then we will prove $\mathcal{E}(\pi)$ is $\alpha$-convex, i.e.
    for $P_0,P_1\in \mathcal{E}(\pi)$, $\alpha>0$ and $\lambda\in (0,1)$, there exists 
    \begin{equation*}
        P_{\lambda}(x,y):=\dfrac{\left((1-\lambda) P_0^{\alpha}(x,y)+\lambda P_1^{\alpha}(x,y)\right)^{\frac{1}{\alpha}}}{\sum_{x,y}\left((1-\lambda) P_0^{\alpha}(x,y)+\lambda P_1^{\alpha}(x,y)\right)^{\frac{1}{\alpha}}}\in \mathcal{E}(\pi).
    \end{equation*}
    Actually, we only need to verify that $P_{\lambda}(x,y)$ is symmetric, after which $G_{\lambda}(x,y):=P_{\lambda}(x,y)/\pi(x)$ is reversible and hence satisfies the condition for $\mathcal{E}(\pi)$. Actually, we have
    \begin{align*}
        P_{\lambda}(x,y)&=\dfrac{\left((1-\lambda) P_0^{\alpha}(x,y)+\lambda P_1^{\alpha}(x,y)\right)^{\frac{1}{\alpha}}}{\sum_{x,y}\left((1-\lambda) P_0^{\alpha}(x,y)+\lambda P_1^{\alpha}(x,y)\right)^{\frac{1}{\alpha}}}\\
        &=\dfrac{\left((1-\lambda) P_0^{\alpha}(y,x)+\lambda P_1^{\alpha}(y,x)\right)^{\frac{1}{\alpha}}}{\sum_{x,y}\left((1-\lambda) P_0^{\alpha}(y,x)+\lambda P_1^{\alpha}(y,x)\right)^{\frac{1}{\alpha}}}\\
        &=P_{\lambda}(y,x),
    \end{align*}
    where the second equality follows from the symmetry of $P_0,P_1$. Moreover, since $\pi(x)>0$, we can define a bijection $\phi\in \left(\mathcal{L}(\pi),\mathcal{E}(\pi)\right): M(x,y)\mapsto \pi(x)M(x,y)$, and therefore we have
    \begin{align*}
        &R_{\alpha}(M\|L)\geq R_{\alpha}\left(L^{(R_\alpha)}\|L\right), \quad \forall M\in \mathcal{L}(\pi)\\
        \iff & \widetilde R_{\alpha}\left(P\|\widetilde L\right)\geq \widetilde R_{\alpha}\left(\widetilde{L^{(R_\alpha)}}\|L\right), \quad \forall P\in \mathcal{E}(\pi),
    \end{align*}
    which implies $\widetilde{L^{(R_\alpha)}}\in \mathop{\arg\min}_{P\in \mathcal{E}(\pi)}\widetilde R_{\alpha}\left(P\|\widetilde L\right)$. Thus by \parencite[Theorem 14]{van2014renyi}, for $M\in \mathcal{L}(\pi)$ and $\widetilde M(x,y)=\pi(x)M(x,y)\in \mathcal{E}(\pi)$, we have
    \begin{align*}
        R_{\alpha}(M\|L)&=\widetilde R_{\alpha}\left(\widetilde M\|\widetilde L\right)\geq R_{\alpha}\left(\widetilde M\|\widetilde{L^{(R_\alpha)}}\right)+R_{\alpha}\left(\widetilde{L^{(R_\alpha)}}\|\widetilde L\right)\\
        &=R_{\alpha}(M\|L^{(R_\alpha)})+R_{\alpha}(L^{(R_\alpha)}\|L),
    \end{align*}
    which is \eqref{pythagorean inequality}.
    
    To prove the weak triangle inequality, by Cauchy-Schwarz inequality we have 
    \begin{align*}
        &\sum_{x,y}\pi(x)U^{\alpha}(x,y)V^{1-\alpha}(x,y)\\
       =&\sum_{x,y}\pi(x) W(x,y)\left(\dfrac{U(x,y)}{W(x,y)}\right)^{\alpha}\left(\dfrac{W(x,y)}{V(x,y)}\right)^{\alpha-1}\\
       \leq & \left(\sum_{x,y}\pi(x) W(x,y)\left(\dfrac{U(x,y)}{W(x,y)}\right)^{2\alpha}\right)^{1/2}
       \left(\sum_{x,y}\pi(x) W(x,y)\left(\dfrac{W(x,y)}{V(x,y)}\right)^{2\alpha-2}\right)^{1/2},
    \end{align*}
    taking logarithm we have
    \begin{equation*}
        (\alpha-1)R_{\alpha}(U\|V)\leq \dfrac{2\alpha-1}{2}R_{2\alpha}(U\|W)+\dfrac{2\alpha-2}{2}R_{2\alpha-1}(W\|V),
    \end{equation*}
    which is \eqref{weak triangle inequality}.
\end{proof}

\begin{remark}
    Using Proposition \ref{prop:pythagorean inequality and weak triangle inequaliy}, the uniqueness of $\alpha$-projection can be easily verified. Suppose $L_1^{(R_\alpha)}$ and $L_2^{(R_\alpha)}$ are two $\alpha$-projections, then by \eqref{pythagorean inequality} we have
    \begin{equation*}
        R_{\alpha}\left(L_2^{(R_\alpha)}\|L\right)\geq R_{\alpha}\left(L_2^{(R\alpha)}\|L_1^{(R_\alpha)}\right)+R_{\alpha}\left(L_1^{(R_\alpha)}\|L\right),
    \end{equation*}
    which implies $R_{\alpha}\left(L_2^{(R_\alpha)}\|L_1^{(R_\alpha)}\right)=0$ and therefore $L_1^{(R_\alpha)}=L_2^{(R_\alpha)}$ by Proposition \ref{prop:non-negative}.
\end{remark}

The next result introduces the notion of Chernoff information of Markov chains:

\begin{proposition}[Chernoff information]\label{Chernoff information: informal}
    
    For $\alpha > 0$, $M,L\in \mathcal{L}$ and $\pi \in \mathcal{P}(\mathcal{X})$, we have
    \begin{equation}\label{eq:tradeoff of KL}
        (1-\alpha)R_{\alpha}(M\|L)=\inf_{P\in \mathcal{L}}\left\{\alpha D_{KL}(P\|M)+(1-\alpha)D_{KL}(P\|L)\right\}.
    \end{equation}
    Moreover, if $D_{KL}(M\|L)<\infty$, we have the saddle property
    \begin{align}
        &\sup_{\alpha>0}\inf_{P\in \mathcal{L}}\left\{\alpha D_{KL}(P\|M)+(1-\alpha)D_{KL}(P\|L)\right\}\label{chernoff information}\\
        =&\inf_{P\in \mathcal{L}}\sup_{\alpha>0}\left\{\alpha D_{KL}(P\|M)+(1-\alpha)D_{KL}(P\|L)\right\},\nonumber
    \end{align}
    where we call \eqref{chernoff information} as the Chernoff information with respect to $M,L$.
\end{proposition}
\begin{proof}
    Following the proof in Proposition \ref{prop:pythagorean inequality and weak triangle inequaliy}, we denote the edge measure family 
    \begin{equation*}
    \mathcal{P}_e:=\left\{\pi(x)g(x,y):g(x,y)\geq 0, \sum_{y}g(x,y)=1\right\},
    \end{equation*}
    then we have $\mathcal{E}(\pi)\subset \mathcal{P}_e$, and there exists a bijection from $\mathcal{L}$ to $\mathcal{P}_e$. Therefore, by \parencite[Theorem 30]{van2014renyi} we have 
    \begin{align*}
        & \inf_{P\in\mathcal{L}}\left\{\alpha D_{KL}(P\|M)+(1-\alpha)D_{KL}(P\|L)\right\}\\
        =&\inf_{P\in\mathcal{L}}\left\{\alpha \widetilde R_1\left(\widetilde P\|\widetilde M\right)+(1-\alpha)\widetilde R_1\left(\widetilde P\|\widetilde L\right)\right\}\\
        =&\inf_{Q\in\mathcal{P}_e}\left\{\alpha \widetilde R_1\left(Q\|\widetilde M\right)+(1-\alpha)\widetilde R_1\left(Q\|\widetilde L\right)\right\}\\
        =&(1-\alpha)\widetilde R_{\alpha}\left(\widetilde M\|\widetilde L\right)\\
        =&(1-\alpha)R_{\alpha}(M\|L).
    \end{align*}
    For the saddle property, the proof is similar by using \parencite[Theorem 32]{van2014renyi}.
\end{proof}

In Section \ref{sec:hypothesis testing}, we will illustrate the role of Chernoff information in estimating the Bayesian error of binary hypothesis testing between two transition matrices.

The final result of this section discusses some properties of the ergodicity or contraction coefficients as introduced in Definition \ref{def:ergodicity coefficients}. In particular, these results offer a bridge and connect $\eta^f, \eta^{(R_{\alpha})}$ with their counterparts such as $\widetilde \eta^f, \widetilde \eta^{(R_{\alpha})}$, thus transferring properties from one to another. We shall also see in the proof that many of the properties of these coefficients are inherited from analogous properties of their respective information divergences such as $D_f$ or $R_{\alpha}$.

\begin{proposition}
    [Ergodicity coefficients]\label{prop:ergocoef}

    Given $P\in\mathcal{L}$, $\pi\in\mathcal{P}(\mathcal{X})$, $\alpha\in (0,1)\cup (1,+\infty)$ and $f$ being a convex function on $\mathbb{R}$ with $f(1) = 0$, we have
    \begin{enumerate}[label=(\roman*).]
        \item(Non-negativity and normalization)\label{it:ergcoef1}
        
        $0\leq \eta^f(P)\leq \eta^{(\mathrm{TV})}(P)\leq 1$, $0\leq \eta^{(R_{\alpha})}(P)\leq 1$ and $0\leq \widetilde \eta^{(R_{\alpha})}(P)\leq 1$.

        \item(Ergodicity coefficients on the space $\mathcal{P}\label{it:ergcoef2}
        (\mathcal{X})$ are the same as that on the space $\mathcal{L}$)
        
        $\eta^f(P)=\widetilde \eta^f(P)$ and $\eta^{(R_\alpha)}(P) \geq \widetilde \eta^{(R_\alpha)}(P)$.

        \item(Bounds)\label{it:ergcoef3}
        Suppose that the sequence $(U^*_n,V^*_n)$ approaches $\eta^{(R_{\alpha})}(P)$, and the sequence $(\mu^*_n,\nu^*_n)$ approaches $\widetilde \eta^{R_{\alpha}}(P)$, that is,
        \begin{align*}
            \eta^{(R_{\alpha})}(P) &= \lim_{n \to \infty} \dfrac{R_{\alpha}(U^*_nP\|V^*_nP)}{R_{\alpha}(U^*_n\|V^*_n)}, \quad
            \widetilde \eta^{(R_{\alpha})}(P) = \lim_{n \to \infty} \dfrac{\widetilde R_{\alpha}(\mu^*_n P\|\nu^*_n P)}{\widetilde R_{\alpha}(\mu^*_n\|\nu^*_n)},
        \end{align*}
        and by passing to a subsequence if necessary, we assume that $(U^*_n,V^*_n) \rightarrow (U^*,V^*)$ and $(\mu^*_n,\nu^*_n) \rightarrow (\mu^*,\nu^*)$ pointwise. We then have, for $\alpha > 1$,
        \begin{align*}
            \eta^{(D_\alpha)}(P) &\leq \eta^{(R_\alpha)}(P) \leq \begin{cases}
            \dfrac{\ln(1+(\alpha-1)\eta^{(D_\alpha)}(P)D_\alpha(U^*\|V^*))}{\ln(1+(\alpha-1)D_\alpha(U^*\|V^*))} \leq 1, \quad &\textrm{if }\, U^* \neq V^*,\\
            \eta^{(D_\alpha)}(P) \leq 1, \quad &\textrm{if }\, U^* = V^*,
            \end{cases}    \\
            \widetilde \eta^{(D_\alpha)}(P) &\leq \widetilde \eta^{(R_\alpha)}(P) \leq 
            \begin{cases}
            \dfrac{\ln(1+(\alpha-1)\widetilde \eta^{(D_\alpha)}(P)\widetilde D_\alpha(\mu^*\|\nu^*))}{\ln(1+(\alpha-1)\widetilde D_\alpha(\mu^*\|\nu^*))} \leq 1, \quad &\textrm{if }\, \mu^* \neq \nu^*,\\
            \widetilde \eta^{(D_\alpha)}(P) \leq 1, \quad &\textrm{if }\, \mu^* = \nu^*,
            \end{cases}    
        \end{align*}
        while for $\beta \in (0,1)$, we have
        \begin{align*}
            \eta^{(R_\beta)}(P) &\leq \eta^{(D_\beta)}(P), \quad \widetilde \eta^{(R_\beta)}(P) \leq \widetilde \eta^{(D_\beta)}(P). 
        \end{align*}
        
        \item(Convexity)\label{it:ergcoef4}
        
        The mappings $\mathcal{L} \ni P\mapsto \eta^f(P)$ and $\mathcal{L} \ni P\mapsto \eta^{(R_\alpha)}(P)$ are convex if $\alpha \in (0,1)$.

        \item(Submultiplicativity)\label{it:submulti}

        For $P,Q \in \mathcal{L}$, we have
        $$\eta^f(PQ) \leq \eta^f(P) \eta^f(Q).$$

        \item(Symmetry)\label{it:ergcoef5}
        
        If $\alpha\in (0,1)$, then $\eta^{(R_\alpha)}(P)=\eta^{(R_{1-\alpha})}(P)$ and $\eta^{(D_\alpha)}(P)=\eta^{(D_{1-\alpha})}(P)$.

        \item(Equivalence)\label{it:ergcoef6}
        
        For $0 < \alpha < \beta < 1$, we have
        $$\dfrac{\alpha}{\beta}\dfrac{1-\beta}{1-\alpha} \eta^{(R_\beta)}(P)\leq \eta^{(R_\alpha)}(P)\leq \dfrac{\beta}{\alpha}\dfrac{1-\alpha}{1-\beta} \eta^{(R_\beta)}(P).$$
        
        \item(Scrambling)\label{it:ergcoef7}
        
        $P$ is scrambling if and only if any one (and hence all) of the following ergodicity coefficients is less than $1$:
        $$\eta^f(P),\eta^{(D_{\alpha})}(P), \eta^{(R_{\beta})}(P), \widetilde \eta^{(D_{\alpha})}(P), \widetilde \eta^{(R_{\beta})}(P) < 1,$$
        where $\beta > 1$.

        \item(Unit rank and independence)\label{it:ergcoef8}
        
        Let $f$ be strictly convex at $1$. Then $P$ has unit rank if and only if any one (and hence all) of the following ergodicity coefficients equals to $0$:
        $$\eta^f(P),\eta^{(D_{\alpha})}(P), \eta^{(R_{\beta})}(P), \widetilde \eta^{(D_{\alpha})}(P), \widetilde \eta^{(R_{\beta})}(P) = 0,$$
        where $\beta > 1$.
    \end{enumerate}
\end{proposition}

\begin{proof}
    Item \ref{it:ergcoef1}: Using item \ref{it:ergcoef2} and the property that $0\leq \widetilde \eta^f(P)\leq \widetilde \eta^{(\mathrm{TV})}(P)\leq 1$ from \parencite[Theorem 4.1]{cohen1993relative}, we see that $\eta^f(P) \in [0,1]$. By item \ref{it:ergcoef3}, $\eta^{(R_{\alpha})}(P), \widetilde \eta^{(R_{\alpha})}(P) \in [0,1]$.
    
    Item \ref{it:ergcoef2}: We first prove that $\widetilde \eta^f(P) \geq \eta^f(P)$. Fix $M,L,P \in \mathcal{L}$. Using \eqref{eq:weightedDfML} twice, we see that
    \begin{align*}
        D_f(MP\|LP) &= \sum_{x \in \mathcal{X}} \pi(x) \widetilde D_f(MP(x,\cdot)\|LP(x,\cdot)) \\
                    &\leq \sum_{x \in \mathcal{X}} \pi(x) \widetilde{\eta}_f(P) \widetilde D_f(M(x,\cdot)\|L(x,\cdot)) \\
                    &= \widetilde{\eta}^f(P) D_f(M\|L),
    \end{align*}
    where the inequality follows from the definition of $\widetilde{\eta}^f(P)$. Taking supremum over $M,L$ with $M \neq L$ yields the desired result.

    Next, we prove the opposite inequality that $\widetilde \eta^f(P) \leq \eta^f(P)$ and $\widetilde \eta^{(R_{\alpha})}(P) \leq \eta^{(R_{\alpha})}(P)$. Fix $\mu,\nu \in \mathcal{P}(\mathcal{X})$, and let $U(x,\cdot) = \mu, V(x,\cdot) = \nu$ for all $x \in \mathcal{X}$. Using the definitions of $D_f$ and $\widetilde D_f$, we thus see that
    $$D_f(UP \| VP) = \widetilde D_f(\mu P\| \nu P),$$
    which yields
    $$R_{\alpha}(UP \| VP) = \widetilde R_{\alpha}(\mu P\| \nu P),$$
    and hence
    \begin{align*}
        \eta^f(P) &\geq \sup_{\mu \neq \nu; U(x,\cdot) = \mu, V(x,\cdot) = \nu} \dfrac{D_f(UP\|VP)}{D_f(U\|V)} = \sup_{\mu \neq \nu} \dfrac{\widetilde D_f(\mu P\|\nu P)}{\widetilde D_f(\mu\|\nu)} = \widetilde \eta^f(P), \\
        \eta^{(R_{\alpha})}(P) &\geq \sup_{\mu \neq \nu; U(x,\cdot) = \mu, V(x,\cdot) = \nu} \dfrac{R_{\alpha}(UP\|VP)}{R_{\alpha}(U\|V)} = \sup_{\mu \neq \nu} \dfrac{\widetilde R_{\alpha}(\mu P\|\nu P)}{\widetilde R_{\alpha}(\mu\|\nu)} = \widetilde \eta^{(R_{\alpha})}(P).
    \end{align*}
    Therefore we have $\widetilde \eta^f(P)=\eta^f(P)$ and $\eta^{(R_\alpha)}(P) \geq \widetilde \eta^{(R_\alpha)}(P)$.

    Item \ref{it:ergcoef3}: We first prove the second case. For any $\mu,\nu \in \mathcal{P}(\mathcal{X})$, we see that
    \begin{align*}
        \widetilde R_{\alpha}(\mu P \| \nu P) &= \dfrac{1}{\alpha - 1} \ln(1+(\alpha-1) \widetilde D_{\alpha}(\mu P \| \nu P)) \\
        &\leq \widetilde \eta^{(R_\alpha)}(P) \dfrac{1}{\alpha - 1} \ln(1+(\alpha-1) \widetilde D_{\alpha}(\mu  \| \nu )) \\
        &\leq \dfrac{1}{\alpha - 1} \ln(1+\widetilde \eta^{(R_\alpha)}(P) (\alpha-1) \widetilde D_{\alpha}(\mu  \| \nu )),
    \end{align*}
    where the first inequality follows from the definition of $\widetilde \eta^{(R_\alpha)}(P)$ and the second inequality stems from $\alpha > 1$ and the inequality 
    \begin{align}\label{eq:pfinequality}
        (1+x)^a &\leq 1+ax
    \end{align} 
    for $x > -1$ and $a \in [0,1]$. Rearranging the above inequalities give
    \begin{align*}
        \widetilde D_{\alpha}(\mu P \| \nu P) &\leq \widetilde \eta^{(R_\alpha)}(P) \widetilde D_{\alpha}(\mu  \| \nu ),
    \end{align*}
    and hence
    $$\widetilde \eta^{(D_\alpha)}(P) \leq \widetilde \eta^{(R_\alpha)}(P).$$
    
    On the other hand, using the definition of $\widetilde \eta^{(D_\alpha)}(P)$, we see that
    \begin{align*}
        \widetilde D_{\alpha}(\mu^*_n P \| \nu^*_n P)) \leq \widetilde \eta^{(D_\alpha)}(P) \widetilde D_{\alpha}(\mu^*_n \| \nu^*_n),
    \end{align*}
    and applying the non-decreasing transformation $t \mapsto \frac{1}{\alpha-1}\ln(1+(\alpha-1)t)$ gives
    \begin{align*}
        \widetilde R_{\alpha}(\mu^*_n P \| \nu^*_n P)) \leq \dfrac{1}{\alpha-1}\ln(1+(\alpha-1)\widetilde \eta^{(D_\alpha)}(P) \widetilde D_{\alpha}(\mu^*_n \| \nu^*_n)).
    \end{align*}
    Now, we divide both sides by $\widetilde R_{\alpha}(\mu^*_n \| \nu^*_n ) > 0$ and pass to $n \to \infty$ to yield
    $$\widetilde \eta^{(R_\alpha)}(P) \leq \lim_{n \to \infty} \dfrac{\ln(1+(\alpha-1)\widetilde \eta^{(D_\alpha)}(P)\widetilde D_\alpha(\mu^*_n\|\nu^*_n))}{\ln(1+(\alpha-1)\widetilde D_\alpha(\mu^*_n\|\nu^*_n))} \leq 1.$$
    The remaining inequalities can be obtained analogously by simply replacing $\mu^*_n,\nu^*_n$ by $U_n,V_n$ respectively and $\widetilde D_{\alpha},\widetilde R_{\alpha}$ by their untilded counterparts.
    
    For the case of $\beta \in (0,1)$, similar to the reasoning above we obtain that
    \begin{align*}
        \widetilde \eta^{(R_\beta)}(P) \leq 
            \begin{cases}
            \dfrac{\ln(1+(\beta-1)\widetilde \eta^{(D_\beta)}(P)\widetilde D_\beta(\mu^*\|\nu^*))}{\ln(1+(\beta-1)\widetilde D_\beta(\mu^*\|\nu^*))} \leq \widetilde \eta^{(D_\beta)}(P), \quad &\textrm{if }\, \mu^* \neq \nu^*,\\
            \widetilde \eta^{(D_\beta)}(P), \quad &\textrm{if }\, \mu^* = \nu^*,
            \end{cases}  
    \end{align*}
    where we use again \eqref{eq:pfinequality} in the second inequality and $\beta \in (0,1)$. 

    Item \ref{it:ergcoef4}: First, we note that he mapping $P\mapsto \eta^f(P) = \widetilde \eta^f(P)$ is convex according to \parencite[Section 4]{cohen1993relative}. Next, for $P,M \neq L \in \mathcal{L}$, the ratio 
    \begin{align*}
        (M,L) \mapsto \dfrac{R_{\alpha}(MP\|LP)}{R_{\alpha}(M\|L)} 
    \end{align*}
    is jointly convex when $\alpha \in (0,1)$ by Proposition \ref{prop:convexity}, and taking the supremum over $M \neq L$ preserves its convexity.

    Item \ref{it:submulti}: It follows readily from item \ref{it:ergcoef2} and the submultiplicativity of $\widetilde \eta^f$ as in \parencite[Section $4$]{cohen1993relative}.
    
    Item \ref{it:ergcoef5}: Thanks to the skew symmetry property of $R_{\alpha}$ as in Proposition \ref{prop:skewsym}, we see that for $P,M \neq L \in \mathcal{L}$, 
    \begin{align*}
        \dfrac{R_{\alpha}(MP\|LP)}{R_{\alpha}(M\|L)} &= \dfrac{R_{1-\alpha}(LP\|MP)}{R_{1-\alpha}(L\|M)},
    \end{align*}
    which yields the desired result by taking the supremum over $M \neq L \in \mathcal{L}$. The analogous result for $D_{\alpha}$ holds due to the skew symmetry of $D_{\alpha}$, that is, $D_{\alpha}(M\|L) = \frac{\alpha}{1-\alpha} D_{\alpha}(L\|M)$.

    Item \ref{it:ergcoef6}: It follows directly from the equivalence property of $R_{\alpha}$ in Proposition \ref{prop:skewsym}.

    Item \ref{it:ergcoef7}: If $P$ is scrambling, then according to \parencite[Theorem 4.2]{cohen1993relative}, $\eta^f(P) < 1$. In particular, $\eta^{(D_{\alpha})}(P) = \widetilde \eta^{(D_{\alpha})}(P) < 1$ by item \ref{it:ergcoef2}. By item \ref{it:ergcoef4}, this yields $\eta^{(R_{\beta})}(P), \widetilde \eta^{(R_{\beta})}(P) < 1$. On the other hand, if $P$ is not scrambling, then by \parencite[Theorem 4.2]{cohen1993relative} again, $\eta^f(P) = 1$. In particular, $\eta^{(D_{\alpha})}(P) = \widetilde \eta^{(D_{\alpha})}(P) = 1$ by item \ref{it:ergcoef2}. Using again item \ref{it:ergcoef4}, this yields $\eta^{(R_{\beta})}(P), \widetilde \eta^{(R_{\beta})}(P) = 1$. 

    Item \ref{it:ergcoef8}: Using \parencite[Section 4]{cohen1993relative}, unit rank of $P$ is equivalent to $\eta^f(P) = 0$. The rest of the proof is similar to that of item \ref{it:ergcoef7} and is omitted.
\end{proof}

\begin{remark}\label{remark: scrambling}
    In Proposition \ref{prop:ergocoef} item \ref{it:ergcoef7}, a transition matrix $P \in \mathcal{L}$ is called scrambling if for each pair of rows $(x,y)\in \mathcal{X}\times\mathcal{X}$, there exists a column $z\in \mathcal{X}$ such that $P(x,z), P(y,z)>0$, i.e.
    \begin{equation*}
        \min_{x,y\in \mathcal{X}} \sum_{z\in\mathcal{X}} P(x,z)\wedge P(y,z)>0,
    \end{equation*}
    and this means that every two rows of $P$ are not orthogonal, see \parencite[Theorem 3.1]{seneta2006non}.
\end{remark}

In Section \ref{cesaro mixing time}, we will use the ergodicity coefficients to give the relationship between Ces\`aro mixing time and classical mixing time.

\section{Applications}

The main aim of this section is to provide applications of these notions of information divergences and ergodicity coefficients in the context of Markov chain mixing and binary hypothesis testing.

\subsection{Binary hypothesis testing between Markov chains}\label{sec:hypothesis testing}

Given $n$ pairs of i.i.d. data $\left\{(X_1,Y_1),(X_2,Y_2),...,(X_n,Y_n)\right\}$ which are observed from an arbitrary inital distribution $\pi \in \mathcal{P}(\mathcal{X})$ and transition matrix $P$, i.e. $X_i\sim \pi$ and $Y_i|X_i\sim P(X_i,\cdot)$. Note that in this section $\pi \in \mathcal{P}(\mathcal{X})$ is given and needs not be the stationary distribution of $P$. Then, we consider the following hypothesis testing problem between transition matrices:
\begin{equation}
    H_0: P=P_0 \quad\longleftrightarrow \quad H_1:P=P_1, \quad \pi \enspace \text{is known}.\label{hypothesis testing: original}
\end{equation}
Denote $Z_i:=(X_i,Y_i)\in \mathcal{X}\times \mathcal{X}$, therefore $Z_i$ follows the edge measure $Z_i\sim Q(x,y):=\pi(x)P(x,y)$, and the hypothesis testing \eqref{hypothesis testing: original} is equivalent to the following hypothesis testing between edge measures:
\begin{equation}
    H_0':Q=Q_0 \quad\longleftrightarrow \quad H_1':Q=Q_1,\label{hypothesis testing: new}
\end{equation}
where $Q_i(x,y)=\pi(x)P_i(x,y)$ for $i=0,1$.

Now, let $A_n\subseteq (\mathcal{X}\times\mathcal{X})^n$ be the region of rejecting $H_0$ or $H_0'$, and $A_n^c$ be the region of accepting $H_0$ or $H_0'$, then the first and second type error are
\begin{equation}
\alpha_n:=Q_0^n(A_n), \quad \beta_n:=Q_1^n(A_n^c),
\end{equation}
where $Q_0^n$ and $Q_1^n$ are extended measures from $\mathcal{X}\times\mathcal{X}$ to $(\mathcal{X}\times\mathcal{X})^n$. Using \parencite[Theorem 11.7.1]{cover2012elements} to apply Neyman-Pearson lemma on \eqref{hypothesis testing: new}, we can find that the optimal test for both \eqref{hypothesis testing: original} and \eqref{hypothesis testing: new} has the form of likelihood ratio test, i.e.
\begin{equation}
    A_n=A_n(T)=\left\{\bold z=(z_1,z_2,...,z_n)\in (\mathcal{X}\times \mathcal{X})^n: \dfrac{Q_0(z_1)Q_0(z_2)...Q_0(z_n)}{Q_1(z_1)Q_1(z_2)...Q_1(z_n)}<T\right\},
\end{equation}
where there is a slight abuse of notation regarding $Q_0$ and $Q_1$.
Then, we have the following asymptotic equipartition property (AEP):
\begin{theorem}
    Let $\{Z_i=(X_i,Y_i)\}_{i=1}^n$ be i.i.d. random variables drawn from $Q_0(x,y)=\pi(x)P_0(x,y)$. For another edge measure $Q_1(x,y)=\pi(x)P_1(x,y)$, we have
    \begin{equation}
        \dfrac{1}{n} \ln \dfrac{Q_0(Z_1)Q_0(Z_2)...Q_0(Z_n)}{Q_1(Z_1)Q_1(Z_2)...Q_1(Z_n)}\xrightarrow{a.s.} R_1(P_0\|P_1)\quad as \enspace n\rightarrow \infty,
    \end{equation}
    where we recall that $R_1(P_0\|P_1)$ is the $\mathrm{KL}$ divergence from $P_1$ to $P_0$ as defined in Proposition \ref{extensions to alpha=0,1,infty}.
\end{theorem}
\begin{proof}
    By strong law of large numbers, recalling the definition of $R_1(P_0\|P_1)$ defined in Proposition \ref{extensions to alpha=0,1,infty}, we have
    \begin{align*}
        \dfrac{1}{n} \ln \dfrac{Q_0(Z_1)Q_0(Z_2)...Q_0(Z_n)}{Q_1(Z_1)Q_1(Z_2)...Q_1(Z_n)}
        &= \dfrac{1}{n}\sum_{i=1}^n \ln \dfrac{Q_0(Z_i)}{Q_1(Z_i)} \\
        &\xrightarrow{a.s.} \mathbb{E}_{Z_1 \sim Q_0}\left[\ln \dfrac{Q_0(Z_1)}{Q_1(Z_1)}\right]\\
        &=\sum_{x,y}\pi(x)P_0(x,y)\ln \dfrac{P_0(x,y)}{P_1(x,y)}
        =R_1(P_0\|P_1).
    \end{align*}
\end{proof}

Next, we will consider the Bayesian case. Suppose $P=P_0$ has a prior probability $\pi_0$ and $P=P_1$ has a prior probability $\pi_1$, where $\pi_0+\pi_1=1$ and $0<\pi_0<1$. Then, the overall Bayesian error for both \eqref{hypothesis testing: original} and \eqref{hypothesis testing: new} is
\begin{equation}
    P_e(A_n)=\pi_0\alpha_n+\pi_1\beta_n,\label{Bayesian error}
\end{equation}
and we need to estimate the best achievable Bayesian error
\begin{align}
    P_e^*(n)&:=\inf_{A_n\subseteq (\mathcal{X}\times\mathcal{X})^n} P_e(A_n)\label{best achievable Bayesian error}\\
    &=\inf_{A_n\subseteq (\mathcal{X}\times\mathcal{X})^n} \left\{\pi_0 Q_0^n(A_n)+\pi_1Q_1^n(A_n^c)\right\}.\nonumber
\end{align}

Before presenting the estimation of $P_e^*(n)$, we will first give a rigorous definition of Chernoff information between transition matrices $P_0$ and $P_1$:
\begin{align}
    C(P_0,P_1)&:=\max_{0\leq \alpha\leq 1} (1-\alpha)R_{\alpha}(P_0\|P_1)\label{chernoff information: formal}\\
    &=\max_{0\leq \alpha\leq 1}-\ln \sum_{x,y}\pi(x)P_0^{\alpha}(x,y)P_1^{1-\alpha}(x,y).\label{chernoff information:expansion}
\end{align}
When $\alpha>1$, $(1-\alpha)R_{\alpha}(P_0\|P_1)\leq 0$, hence \eqref{chernoff information: formal} is in accordance with  \ref{chernoff information} in Proposition \ref{Chernoff information: informal} using \eqref{eq:tradeoff of KL}, that is,
\begin{align*}
    \sup_{\alpha>0}\inf_{P\in \mathcal{L}}\left\{\alpha D_{KL}(P\|M)+(1-\alpha)D_{KL}(P\|L)\right\}&=\sup_{\alpha>0}(1-\alpha)R_{\alpha}(P_0\|P_1)\\
    &=\max_{0\leq \alpha\leq 1} (1-\alpha)R_{\alpha}(P_0\|P_1).
\end{align*}
Then, we have the following result using Chernoff information.
\begin{theorem}
    Recalling the best achievable Bayesian error defined in \eqref{best achievable Bayesian error}, we have 
    \begin{equation}
        \lim_{n\rightarrow \infty}-\dfrac{1}{n}\ln P_e^*(n)=C(P_0,P_1),
    \end{equation}
\end{theorem}
\begin{proof}
    According to \parencite[Theorem 11.9.1]{cover2012elements}, denote 
    \begin{equation*}
        Q_{\lambda}(x,y):=\dfrac{Q_0^{\lambda}(x,y)Q_1^{1-\lambda}(x,y)}{\sum_{x,y}Q_0^{\lambda}(x,y)Q_1^{1-\lambda}(x,y)},
    \end{equation*}
    then we have 
    \begin{align*}
        \lim_{n\rightarrow \infty}-\dfrac{1}{n}\ln P_e^*(n)&=\max_{0\leq \lambda\leq 1}\min \left\{\widetilde R_1(Q_\lambda\|Q_0), \widetilde R_1(Q_\lambda\|Q_1)\right\}\\
        &=\max_{0\leq \alpha\leq 1}-\ln \sum_{x,y}Q_0^{\alpha}(x,y)Q_1^{1-\alpha}(x,y).
    \end{align*}
    Since 
    \begin{equation*}
        \sum_{x,y}Q_0^{\alpha}(x,y)Q_1^{1-\alpha}(x,y)=\sum_{x,y}\pi(x)P_0^{\alpha}(x,y)P_1^{1-\alpha}(x,y),
    \end{equation*}
    and by the definition of $C(P_0,P_1)$ in \eqref{chernoff information: formal}, the statement is verified.
\end{proof}

\subsection{Information-theoretic mixing times and their spectral bounds}\label{mixing time to stationary distribution}
For given $\varepsilon>0$ and transition matrix $P\in \mathcal{L}$ which admits $\pi$ as its stationary distribution, we first introduce three notions of average mixing times based upon suitable information divergences. We define the $\pi$-weighted $\alpha$-divergence mixing time, $\pi$-weighted R\'enyi divergence mixing time of the order $\alpha \in (0,1)\cup (1,+\infty)$ and the $\pi$-weighted total variation distance mixing time to be respectively
\begin{align}
    \bar{t}_{\mathrm{mix}}^{(D_\alpha)}(\varepsilon,P)&:=\inf\left\{t>0:D_{\alpha}(P^t\|\Pi)<\varepsilon\right\},\label{t mix D}\\
   \bar{t}_{\mathrm{mix}}^{(R_\alpha)}(\varepsilon,P)&:=\inf\left\{t>0:R_{\alpha}(P^t\|\Pi)<\varepsilon\right\},\label{t mix R} \\
   \bar{t}_{\mathrm{mix}}^{(\mathrm{TV})}(\varepsilon,P)&:=\inf\left\{t>0:\mathrm{TV}(P^t,\Pi)<\varepsilon\right\},\label{t mix TV}
\end{align}
where $\Pi$ is a $|\mathcal{X}|\times |\mathcal{X}|$ matrix with each row being $\pi$, i.e. $\Pi(x,y)=\pi(y)$, for all $x,y \in \mathcal{X}$. These average mixing times can be used to assess the quantitative convergence rate of the transition matrix $P^t$ towards $\Pi$. In the first main result below, we obtain spectral bounds of these mixing times in terms of $\gamma_*$, the absolute spectral gap of $P$, which is defined as 
\begin{equation*}
    \lambda_*:=\max \left\{|\lambda|:\lambda\enspace \text{is an eigenvalue of}\enspace P,\enspace \lambda\neq 1\right\},
\end{equation*}
and we also call $\gamma_*:=1-\lambda_*$ as the absolute spectral gap of $P$. By \parencite[Lemma 12.1]{levin2017markov}, if $P$ is aperiodic and irreducible, then $\gamma_*>0$.

\begin{theorem}\label{upper and lower bound of D_alpha}
    Let $P\in\mathcal{L}(\pi)$ be an irreducible, aperiodic and $\pi$-reversible transition matrix, and we denote $\pi_{\mathrm{min}}:=\min_{x}\pi(x)>0$. For $\alpha\in (0,1)\cup (1,+\infty)$, if $\varepsilon<\min \left\{\frac{1}{2|\alpha-1|}, \frac{4\alpha}{|\alpha-1|}\right\}$, we have
    \begin{align*}
        \bar{t}_{\mathrm{mix}}^{(D_\alpha)}(\varepsilon,P)&\leq \begin{cases}
            \dfrac{1}{\gamma_*}\left(\ln \dfrac{2\alpha}{(\alpha-1)\pi_{\mathrm{min}}}+\ln \dfrac{1}{\varepsilon}\right), \quad \textrm{if }\, \alpha>1,\\
            \dfrac{1}{\gamma_*}\left(\ln \dfrac{8\alpha}{(1-\alpha)\pi_{\mathrm{min}}}+\ln \dfrac{1}{\varepsilon}\right), \quad \textrm{if }\, 0<\alpha<1,
        \end{cases}    
    \end{align*}
    and 
    \begin{align*}
        \bar{t}_{\mathrm{mix}}^{(D_\alpha)}(\varepsilon,P) &\geq \begin{cases}
            \dfrac{1}{2\ln \frac{1}{1-\gamma_*}}\left(\ln \dfrac{\pi_{\mathrm{min}}^2}{8}+\ln\dfrac{1}{\varepsilon}\right), \quad \textrm{if }\, \alpha>1,\\
            \dfrac{1}{2\ln \frac{1}{1-\gamma_*}}\left(\ln \dfrac{\pi_{\mathrm{min}}^2\alpha}{16}+\ln \dfrac{1}{\varepsilon}\right), \quad \textrm{if }\, 0<\alpha<1,
        \end{cases}     
    \end{align*}
    where $\gamma_*$ is the absolute spectral gap of $P$. 
\end{theorem}
\begin{proof}
    $P$ is irreducible and aperiodic implies $\gamma_*>0$.
    According to \parencite[Eq.12.13]{levin2017markov}, we have 
    \begin{equation}
        \left |\dfrac{P^t(x,y)}{\pi(y)}-1\right|\leq \dfrac{(1-\gamma_*)^t}{\sqrt{\pi(x)\pi(y)}}\leq \dfrac{e^{-\gamma_*t}}{\pi_{\text{min}}}.\label{eq:original upper bound}
    \end{equation}
    Then, for $\alpha\in (0,1)\cup (1,+\infty)$, we rewrite $D_{\alpha}(P^t\|\Pi)$ as 
    \begin{align*}
        D_{\alpha}(P^t\|\Pi)&=\dfrac{1}{\alpha-1}\sum_x\pi(x)\sum_y \Pi(x,y)\left(\left(\dfrac{P^t(x,y)}{\Pi(x,y)}\right)^\alpha-1\right)\nonumber\\
        &=\dfrac{1}{\alpha-1}\sum_x\pi(x)\sum_y\pi(y)\left(\left(\dfrac{P^t(x,y)}{\pi(y)}\right)^\alpha-1\right).\label{eq:rewrite D(Pt,pi)}
    \end{align*}
    Note that for $x\geq-1$, when $\alpha>1$, $\left|(1+x)^\alpha-1\right|\leq (1+|x|)^\alpha-1$. Moreover, for $x,y>0$, it can be easily verified that $\left(1+\frac{x}{y}\right)^y<e^x$. Therefore, using \eqref{eq:original upper bound}, for $\alpha>1$, we have
    \begin{align*}
        D_{\alpha}(P^t\|\Pi)&=\dfrac{1}{\alpha-1}\sum_x\pi(x)\sum_y\pi(y)\left(\left(1+\dfrac{P^t(x,y)}{\pi(y)}-1\right)^\alpha-1\right)\\
        &\leq \dfrac{1}{\alpha-1}\sum_x\pi(x)\sum_y\pi(y)\left(\left(1+\left|\dfrac{P^t(x,y)}{\pi(y)}-1\right|\right)^\alpha-1\right)\\
        &\leq \dfrac{1}{\alpha-1}\sum_x\pi(x)\sum_y\pi(y)\left(\exp \left(\alpha\left|\dfrac{P^t(x,y)}{\pi(y)}-1\right|\right)-1\right)\\
        &\leq \dfrac{1}{\alpha-1}\left(\exp \left(\dfrac{\alpha}{\pi_{\mathrm{min}}}e^{-\gamma_*t}\right)-1\right).
    \end{align*}
    Therefore, in order to have $D_{\alpha}(P^t\|\Pi)<\varepsilon$, we have
    \begin{align*}
        \bar{t}_{\mathrm{mix}}^{(D_\alpha)}(\varepsilon,P)&\leq \dfrac{1}{\gamma_*}\left(\ln \dfrac{\alpha}{\pi_{\mathrm{min}}}+\ln \dfrac{1}{\ln \left(1+(\alpha-1)\varepsilon\right)}\right)\\
        &\leq \dfrac{1}{\gamma_*}\left(\ln \dfrac{2\alpha}{(\alpha-1)\pi_{\mathrm{min}}}+\ln \dfrac{1}{\varepsilon}\right), \quad \alpha>1,
    \end{align*}
    where the last inequality comes from $\varepsilon<\frac{1}{\alpha-1}$ and $\ln (1+x)\geq x-\frac{x^2}{2}\geq \frac{x}{2}$ for $x<1$.
    
    When $0<\alpha<1$, denote $h_t=h_t(x,y):=P^t(x,y)/\pi(y)\geq 0$ and $g_t=g_t(x,y):=1-h_t(x,y)\leq 1$ as shorthand notation. From \eqref{eq:original upper bound}, we directly have for all $x,y\in\mathcal{X}$,
    \begin{equation}
        |g_t(x,y)|\leq \dfrac{e^{-\gamma_*t}}{\pi_{\mathrm{min}}}, \quad 1-\dfrac{e^{-\gamma_*t}}{\pi_{\mathrm{min}}}\leq h_t(x,y) \leq 1+\dfrac{e^{-\gamma_*t}}{\pi_{\mathrm{min}}}.\label{eq:x_t,y_t}
    \end{equation}
    Therefore, for $t>\frac{1}{\gamma_*}\ln \frac{2}{\pi_{\mathrm{min}}}$, we have $\frac{1}{2}\leq h_t(x,y)\leq \frac{3}{2}$, in which case
    \begin{align*}
        D_{\alpha}(P^t\|\Pi)&=\dfrac{1}{\alpha-1}\sum_x\pi(x)\sum_y\pi(y)\left(\left(\dfrac{P^t(x,y)}{\pi(y)}\right)^\alpha-1\right)\\
        &=\dfrac{1}{1-\alpha}\sum_x\pi(x)\sum_y\pi(y)\left(1-h_t^\alpha(x,y)\right)\\
        &=\dfrac{1}{1-\alpha}\sum_x\pi(x)\sum_y\pi(y)\left(\left(g_t(x,y)+h_t(x,y)\right)^\alpha-h_t^\alpha(x,y)\right)\\
        &=\dfrac{1}{1-\alpha}\sum_x\pi(x)\sum_y\pi(y)h_t^{\alpha}\left(\left(1+\dfrac{g_t}{h_t}\right)^{\alpha}-1\right)\\
        &\leq \dfrac{1}{1-\alpha}\sum_x\pi(x)\sum_y\pi(y)h_t^{\alpha}\left(\left(1+\left|\dfrac{g_t}{h_t}\right|\right)^{\alpha}-1\right)\\
        &\leq \dfrac{1}{1-\alpha}\sum_x\pi(x)\sum_y\pi(y)h_t^{\alpha}\left(\exp\left(\alpha \left|\dfrac{g_t}{h_t}\right|\right)-1\right)\\
        &\leq \dfrac{1}{1-\alpha}\left(\dfrac{3}{2}\right)^{\alpha}\left(\exp\left(\dfrac{2\alpha}{\pi_{\mathrm{min}}}e^{-\gamma_*t}\right)-1\right)\\
        &< \dfrac{2}{1-\alpha}\left(\exp\left(\dfrac{2\alpha}{\pi_{\mathrm{min}}}e^{-\gamma_*t}\right)-1\right),
    \end{align*}
    where the second inequality comes from $\left(1+\frac{x}{y}\right)^y<e^x$ for $x,y>0$, the third inequality comes from \eqref{eq:x_t,y_t}, and the last inequality comes from $\alpha<1$. Hence, in order to have the last term above smaller than $\varepsilon$, apart from $t>\frac{1}{\gamma_*}\ln \frac{2}{\pi_{\mathrm{min}}}$, we still need 
    \begin{equation*}
        t>\dfrac{1}{\gamma_*}\left(\ln \dfrac{2\alpha}{\pi_{\mathrm{min}}}+\ln \dfrac{1}{\ln \left(1+\frac{1-\alpha}{2}\varepsilon\right)}\right),
    \end{equation*}
    which implies 
    \begin{align*}
        \bar{t}_{\mathrm{mix}}^{(D_\alpha)}(\varepsilon,P)&\leq \max \left\{\dfrac{1}{\gamma_*}\ln \dfrac{2}{\pi_{\mathrm{min}}}, \dfrac{1}{\gamma_*}\left(\ln \dfrac{2\alpha}{\pi_{\mathrm{min}}}+\ln \dfrac{1}{\ln \left(1+\frac{1-\alpha}{2}\varepsilon\right)}\right)\right\}\\
        &\leq \dfrac{1}{\gamma_*} \max \left\{\ln \dfrac{2}{\pi_{\mathrm{min}}}, \left(\ln \dfrac{8\alpha}{(1-\alpha)\pi_{\mathrm{min}}}+\ln \dfrac{1}{\varepsilon}\right)\right\}\\
        &=\dfrac{1}{\gamma_*}\left(\ln \dfrac{8\alpha}{(1-\alpha)\pi_{\mathrm{min}}}+\ln \dfrac{1}{\varepsilon}\right), \quad 0<\alpha<1,
    \end{align*}
    where the second inequality comes from $\varepsilon<\frac{1}{1-\alpha}$ and $\ln (1+x)\geq \frac{x}{2}$ for $x<1$, and the equality comes from $\varepsilon<\frac{4\alpha}{1-\alpha}$.
    
    On the other hand, we first recall the Pinsker's inequality Proposition \ref{inequalities with other divergences} which implies for $\alpha> 1$,
    \begin{equation}
        \dfrac{1}{2}\mathrm{TV}^2(P^t,\Pi)\leq R_{\alpha}(P^t\|\Pi)=\dfrac{1}{\alpha-1}\ln \left(1+(\alpha-1)D_{\alpha}(P^t\|\Pi)\right),\label{Pinkser inequality for alpha>1}
    \end{equation}
    and for $0<\alpha<1$,
    \begin{equation}
        \dfrac{\alpha}{2}\mathrm{TV}^2(P^t,\Pi)\leq R_{\alpha}(P^t\|\Pi)=\dfrac{1}{\alpha-1}\ln \left(1+(\alpha-1)D_{\alpha}(P^t\|\Pi)\right).\label{Pinsker inequality for alpha<1}
    \end{equation}
    Then we use the statement in \parencite[Remark 12.6]{levin2017markov} that if $\lambda$ is an eigenvalue of $P$ and $\lambda\neq 1$, there exists
    \begin{equation}
        d^{(\mathrm{TV})}(t):=\max_x \dfrac{1}{2}\sum_y\left|P^t(x,y)-\pi(y)\right|\geq \dfrac{1}{2}|\lambda|^t,\label{eq:original lower bound}
    \end{equation}
    which yields
    \begin{equation*}
        \mathrm{TV}(P^t\|\Pi)=\dfrac{1}{2}\sum_{x,y}\pi(x)\left|P^t(x,y)-\pi(y)\right|\geq \pi_{\mathrm{min}}d^{(\mathrm{TV})}(t)\geq \dfrac{\pi_{\mathrm{min}}}{2}|\lambda|^t,
    \end{equation*}
    since $\lambda$ is arbitrary, we have $\mathrm{TV}(P^t\|\Pi)\geq \frac{\pi_{\mathrm{min}}}{2}(1-\gamma_*)^t$. Combining with the inequalities \eqref{Pinkser inequality for alpha>1} and \eqref{Pinsker inequality for alpha<1} above, we obtain that when $\alpha>1$,
    \begin{align*}
        D_{\alpha}(P^t\|\Pi)&\geq \dfrac{1}{\alpha-1}\left(\exp \left(\dfrac{\alpha-1}{2}\mathrm{TV}^2(P^t,\Pi)\right)-1\right)\\
        &\geq \dfrac{1}{\alpha-1}\left(\exp \left(\dfrac{(\alpha-1)\pi_{\mathrm{min}}^2}{8}(1-\gamma_*)^{2t}\right)-1\right),
    \end{align*}
    and when $0<\alpha <1$, similarly we have
    \begin{equation*}
         D_{\alpha}(P^t\|\Pi)\geq \dfrac{1}{\alpha-1}\left(\exp \left(\dfrac{\alpha(\alpha-1)\pi_{\mathrm{min}}^2}{8}(1-\gamma_*)^{2t}\right)-1\right).
    \end{equation*}
    Therefore, we can obtain
    \begin{align*}
        \bar{t}_{\mathrm{mix}}^{(D_\alpha)}(\varepsilon,P)&\geq \dfrac{1}{2\ln \frac{1}{1-\gamma_*}}\left(\ln \dfrac{\pi_{\mathrm{min}}^2(\alpha-1)}{8}+\ln\dfrac{1}{\ln (1+(\alpha-1)\varepsilon)}\right)\\
        &\geq \dfrac{1}{2\ln \frac{1}{1-\gamma_*}}\left(\ln \dfrac{\pi_{\mathrm{min}}^2}{8}+\ln\dfrac{1}{\varepsilon}\right), \quad \alpha>1,
    \end{align*}
    where the second inequality comes from $\ln (1+x)\leq x$. Similary for $0<\alpha<1$, we have 
    \begin{align*}
        \bar{t}_{\mathrm{mix}}^{(D_\alpha)}(\varepsilon,P)&\geq
        \dfrac{1}{2\ln \frac{1}{1-\gamma_*}}\left(\ln \dfrac{\pi_{\mathrm{min}}^2\alpha(1-\alpha)}{8}+\ln \left|\dfrac{1}{\ln (1+(\alpha-1)\varepsilon)}\right|\right)\\
        &\geq \dfrac{1}{2\ln \frac{1}{1-\gamma_*}}\left(\ln \dfrac{\pi_{\mathrm{min}}^2\alpha}{16}+\ln \dfrac{1}{\varepsilon}\right), \quad 0<\alpha<1,
    \end{align*}
    where the second inequality comes from $\ln (1+x)\geq \frac{x}{1+x}$ for $x>-1$ and $\varepsilon<\frac{1}{2(1-\alpha)}$.
\end{proof}

From the definition of $\alpha$-divergence and R\'enyi divergence of Markov chains, we see that 
\begin{equation*}
        R_{\alpha}(P^t\|\Pi)<\varepsilon \iff D_{\alpha}(P^t\|\Pi)<\dfrac{1}{\alpha-1}\left(e^{(\alpha-1)\varepsilon}-1\right),
\end{equation*}
and hence the $\pi$-weighted $\alpha$-divergence mixing time and the R\'enyi divergence mixing time are related via 
    \begin{equation}\label{eq:relationship between t D_alpha and t R_alpha}
       \bar{t}_{\mathrm{mix}}^{(R_\alpha)}(\varepsilon,P)=\bar{t}_{\mathrm{mix}}^{(D_\alpha)}\left(\dfrac{1}{\alpha-1}\left(e^{(\alpha-1)\varepsilon}-1\right),P\right).
    \end{equation}
Thus, to give spectral bounds on $\bar{t}_{\mathrm{mix}}^{(R_\alpha)}(\varepsilon,P)$, using \eqref{eq:relationship between t D_alpha and t R_alpha} yields the following corollary of Theorem \ref{upper and lower bound of D_alpha}:
\begin{corollary}\label{corollary of t D_alpha}
Let $P\in\mathcal{L}(\pi)$ be an irreducible, aperiodic, reversible transition matrix and we denote by $\pi_{\mathrm{min}}:=\min_{x}\pi(x)>0$. For $\alpha\in (0,1)\cup (1,+\infty)$, if $\varepsilon<\min \left\{\frac{1}{2|\alpha-1|}, \frac{4\alpha}{|\alpha-1|}\right\}$, we have
    \begin{align*}
       \bar{t}_{\mathrm{mix}}^{(R_\alpha)}(\varepsilon,P)&\leq 
       \begin{cases}
           \dfrac{1}{\gamma_*}\left(\ln \dfrac{2\alpha}{\pi_{\mathrm{min}}}+\ln \dfrac{1}{e^{(\alpha-1)\varepsilon}-1}\right), &\quad \textrm{if }\, \alpha>1,\\
           \dfrac{1}{\gamma_*}\left(\ln \dfrac{8\alpha}{\pi_{\mathrm{min}}}+\ln \dfrac{1}{1-e^{(\alpha-1)\varepsilon}}\right), &\quad \textrm{if }\, 0<\alpha<1,
       \end{cases}
    \end{align*}
    and 
    \begin{align*}
       \bar{t}_{\mathrm{mix}}^{(R_\alpha)}(\varepsilon,P)&\geq  
       \begin{cases}
           \dfrac{1}{2\ln \frac{1}{1-\gamma_*}}\left(\ln \dfrac{(\alpha-1)\pi_{\mathrm{min}}^2}{8}+\ln\dfrac{1}{e^{(\alpha-1)\varepsilon}-1}\right), &\quad \textrm{if }\, \alpha>1,\\
           \dfrac{1}{2\ln \frac{1}{1-\gamma_*}}\left(\ln \dfrac{\pi_{\mathrm{min}}^2\alpha(1-\alpha)}{16}+\ln \dfrac{1}{1-e^{(\alpha-1)\varepsilon}}\right), &\quad \textrm{if }\, 0<\alpha<1,
       \end{cases}
    \end{align*}
    where $\gamma_*$ is the absolute spectral gap of $P$.    
\end{corollary}

\begin{remark}
    To study the order of $\gamma_*$ in the upper and lower bound of  $\bar{t}_{\mathrm{mix}}^{(D_\alpha)}(\varepsilon,P)$ in Theorem \ref{upper and lower bound of D_alpha} and $\bar{t}_{\mathrm{mix}}^{(R_\alpha)}(\varepsilon,P)$ in Corollary \ref{corollary of t D_alpha}, we can see that if $\gamma_*$ is close to $0$, we have
    \begin{equation*}
        \lim_{\gamma_*\rightarrow 0^+}\dfrac{1}{\gamma_*}\bigg /\dfrac{1}{\ln \frac{1}{1-\gamma_*}}=1,
    \end{equation*}
    which implies the upper bound and lower bound are approximately the same order with respect to $\gamma_*$ in this situation. Indeed, in many statistical physics models or Markov chains of interest, $\gamma_*=\mathcal{O}(1/N^a)$ for some $a>0$ and $N$ being a problem-dependent parameter which is large.  One example is random walk on $N$-dimensional hypercube which will be explained in Example \ref{example:random walk on hypercube}.
\end{remark}

Next, we recall the classical worst-case mixing times based upon $\alpha$-divergence, R\'enyi divergence and total variation distance to be respectively, for $\varepsilon > 0$,
\begin{align}
     t_{\mathrm{mix}}^{(\mathrm{D_{\alpha}})}(\varepsilon,P)&:=\inf\left\{t>0:d^{(D_{\alpha})}(t)<\varepsilon\right\}, \nonumber \\
     t_{\mathrm{mix}}^{(\mathrm{R_{\alpha}})}(\varepsilon,P)&:=\inf\left\{t>0:d^{(R_{\alpha})}(t)<\varepsilon\right\}, \nonumber \\
     t_{\mathrm{mix}}^{(\mathrm{TV})}(\varepsilon,P)&:=\inf\left\{t>0:d^{(\mathrm{TV})}(t)<\varepsilon\right\},\label{t mix worst case}
\end{align}
where the worst-case divergences are defined as, for $t \in \mathbb{N}$,
\begin{align*}
    d^{(D_{\alpha})}(t)&:=\max_{x\in \mathcal{X}}\widetilde D_{\alpha}(P^t(x,\cdot)\|\pi)=\max_{x\in \mathcal{X}}\dfrac{1}{\alpha-1}\left(\sum_{y\in \mathcal{X}}\pi(y)\left(\dfrac{P^t(x,y)}{\pi(y)}\right)^{\alpha}-1\right),\\
    d^{(R_{\alpha})}(t)&:=\max_{x\in \mathcal{X}}\widetilde R_{\alpha}(P^t(x,\cdot)\|\pi)=\max_{x\in \mathcal{X}}\dfrac{1}{\alpha-1}\ln \left(\sum_{y\in \mathcal{X}}\pi(y)\left(\dfrac{P^t(x,y)}{\pi(y)}\right)^{\alpha}\right), \\
    d^{(\mathrm{TV})}(t)&:=\max_{x\in\mathcal{X}}\left\|P^t(x,\cdot)-\pi\right\|_{\mathrm{TV}}=\max_{x\in\mathcal{X}} \dfrac{1}{2}\sum_{y\in\mathcal{X}}\left|P^t(x,y)-\pi(y)\right|.
\end{align*}
We prove the following relationships between these $\pi$-weighted mixing times and and their worst-case mixing times counterpart, and in particular implies that the $\pi$-weighted mixing times and worst-case mixing times are of the same order under some reasonable assumptions on the absolute spectral gap $\gamma_*$:

\begin{theorem}\label{comparison between TV and classical}
    Let $P\in\mathcal{L}(\pi)$ be an irreducible, aperiodic, $\pi$-reversible transition matrix, and we denote $\pi_{\mathrm{min}}:=\min_{x}\pi(x)>0$, then we have the following relationships:
    \begin{enumerate}[label=(\roman*).]
        \item\label{it:averagemix1} $t_{\mathrm{mix}}^{(\mathrm{TV})}\left(\varepsilon,P\right)$ and $\bar{t}_{\mathrm{mix}}^{(\mathrm{TV})}(\varepsilon,P)$: if $\varepsilon<\frac{\pi_{\mathrm{min}}^3}{4}$, then
        \begin{equation*}
            \dfrac{1-\gamma_*}{2}t_{\mathrm{mix}}^{(\mathrm{TV})}\left(\varepsilon,P\right)\leq \bar{t}_{\mathrm{mix}}^{(\mathrm{TV})}(\varepsilon,P)\leq t_{\mathrm{mix}}^{(\mathrm{TV})}(\varepsilon,P).
        \end{equation*}
  
        \item\label{it:averagemix2} $t_{\mathrm{mix}}^{(D_{\alpha})}\left(\varepsilon,P\right)$ and $\bar{t}_{\mathrm{mix}}^{(D_\alpha)}(\varepsilon,P)$: if $\varepsilon<\varepsilon_0:=\min \left\{\frac{1}{2|\alpha-1|}, \frac{4\alpha}{|\alpha-1|}, \frac{|\alpha-1|\pi_{\mathrm{min}}^6}{128\alpha}, \frac{\alpha |\alpha-1|\pi_{\mathrm{min}}^6}{2048}\right\}$, then 
        \begin{equation*}
            \dfrac{1-\gamma_*}{4}t_{\mathrm{mix}}^{(D_\alpha)}(\varepsilon,P)
            \leq \bar{t}_{\mathrm{mix}}^{(D_\alpha)}(\varepsilon,P)\leq t_{\mathrm{mix}}^{(D_\alpha)}(\varepsilon,P), \quad \alpha\in (0,1)\cup (1,\infty).
        \end{equation*}

        \item\label{it:averagemix3} $t_{\mathrm{mix}}^{(R_{\alpha})}\left(\varepsilon,P\right)$ and $\bar{t}_{\mathrm{mix}}^{(R_\alpha)}(\varepsilon,P)$: if $\varepsilon<\frac{1}{\alpha-1}\ln (1+(\alpha-1)\varepsilon_0)$ with $\varepsilon_0$ defined in item \eqref{it:averagemix2}, then 
        \begin{equation*}
            \dfrac{1-\gamma_*}{4}t_{\mathrm{mix}}^{(R_\alpha)}(\varepsilon,P)
            \leq \bar{t}_{\mathrm{mix}}^{(R_\alpha)}(\varepsilon,P)\leq t_{\mathrm{mix}}^{(R_\alpha)}(\varepsilon,P), \quad \alpha\in (0,1)\cup (1,\infty).
        \end{equation*}
    \end{enumerate}
    
\end{theorem}
\begin{proof}
    Item \ref{it:averagemix1}: Recalling the definition of $\mathrm{TV}(P^t,\Pi)$ in Definition \ref{other divergences}, we have 
    \begin{align*}
        \mathrm{TV}(P^t,\Pi)&=\dfrac{1}{2}\sum_x\pi(x)\sum_y\left|P^t(x,y)-\pi(y)\right|\\
        &\leq \sum_x\pi(x)\left(\max_x\dfrac{1}{2}\sum_y\left|P^t(x,y)-\pi(y)\right|\right)\\
        &=d^{(\mathrm{TV})}(t),
    \end{align*}
    and 
    \begin{align*}
        \mathrm{TV}(P^t,\Pi)&=\dfrac{1}{2}\sum_x\pi(x)\sum_y\left|P^t(x,y)-\pi(y)\right|\\
        &\geq \pi_{\mathrm{min}}d^{(\mathrm{TV})}(t),
    \end{align*}
    therefore we can obtain
    \begin{equation}\label{eq:connection between TV and d(t) initial}
        t_{\mathrm{mix}}^{(\mathrm{TV})}\left(\dfrac{\varepsilon}{\pi_{\mathrm{min}}},P\right)\leq \bar{t}_{\mathrm{mix}}^{(\mathrm{TV})}(\varepsilon,P)\leq t_{\mathrm{mix}}^{(\mathrm{TV})}(\varepsilon,P).
    \end{equation}
    According to \parencite[Theorem 12.4, 12.5]{levin2017markov}, there exists
    \begin{equation}\label{eq:upper and lower bound for classical}
        \left(\dfrac{1}{\gamma_*}-1\right)\ln \dfrac{1}{2\varepsilon}
        \leq t_{\mathrm{mix}}^{(\mathrm{TV})}(\varepsilon,P)\leq \dfrac{1}{\gamma_*} \ln \dfrac{1}{\varepsilon \pi_{\mathrm{min}}},
    \end{equation}
    which implies
    \begin{align*}
        t_{\mathrm{mix}}^{(\mathrm{TV})}\left(\dfrac{\varepsilon}{\pi_{\mathrm{min}}},P\right) \geq  \left(\dfrac{1}{\gamma_*}-1\right)\ln \dfrac{\pi_{\mathrm{min}}}{2\varepsilon}
        &\geq (1-\gamma_*)\left(1+\dfrac{\ln \frac{\pi_{\mathrm{min}}^2}{2}}{\ln \frac{1}{\varepsilon}-\ln \pi_{\mathrm{min}}}\right)t_{\mathrm{mix}}^{(\mathrm{TV})}(\varepsilon,P)\\
        &\geq \dfrac{1-\gamma_*}{2}t_{\mathrm{mix}}^{(\mathrm{TV})}(\varepsilon,P),
    \end{align*}
    where the last inequality is from $\varepsilon<\dfrac{\pi_{\mathrm{min}}^3}{4}$. Plugging into \eqref{eq:connection between TV and d(t) initial}, we have
    \begin{equation*}
         \dfrac{1-\gamma_*}{2}t_{\mathrm{mix}}^{(\mathrm{TV})}\left(\varepsilon,P\right)\leq \bar{t}_{\mathrm{mix}}^{(\mathrm{TV})}(\varepsilon,P)\leq t_{\mathrm{mix}}^{(\mathrm{TV})}(\varepsilon,P).
    \end{equation*}
    
    Item \ref{it:averagemix2}: Similar to analysis in item \ref{it:averagemix1}, we have
    \begin{equation*}        t_{\mathrm{mix}}^{(D_\alpha)}\left(\dfrac{\varepsilon}{\pi_{\mathrm{min}}},P\right)\leq \bar{t}_{\mathrm{mix}}^{(D_\alpha)}(\varepsilon,P)\leq t_{\mathrm{mix}}^{(D_\alpha)}(\varepsilon,P).
    \end{equation*}
    Therefore, recalling Theorem \ref{upper and lower bound of D_alpha}, for $\alpha>1$ we have 
    \begin{align*}
        \bar{t}_{\mathrm{mix}}^{(D_\alpha)}(\varepsilon,P)\leq t_{\mathrm{mix}}^{(D_\alpha)}(\varepsilon,P)&\leq \bar{t}_{\mathrm{mix}}^{(D_\alpha)}(\pi_{\mathrm{min}}\varepsilon,P)
        \leq \dfrac{1}{\gamma_*}\left(\ln \dfrac{2\alpha}{(\alpha-1)\pi_{\mathrm{min}}^2}+\ln \dfrac{1}{\varepsilon}\right)\\
        &\leq \bar{t}_{\mathrm{mix}}^{(D_\alpha)}(\varepsilon,P)\dfrac{2\ln \frac{1}{1-\gamma_*}}{\gamma_*}\dfrac{\ln \frac{1}{\varepsilon}+\ln \frac{2\alpha}{(\alpha-1)\pi_{\mathrm{min}}^2}}{\ln \frac{1}{\varepsilon}+\ln \frac{\pi_{\mathrm{min}}^2}{8}}\\
        &\leq \bar{t}_{\mathrm{mix}}^{(D_\alpha)}(\varepsilon,P)\cdot \dfrac{4}{1-\gamma_*}, \quad \alpha>1,
    \end{align*}
    where the last inequality comes from $\ln \frac{1}{1-\gamma_*}\leq \frac{\gamma_*}{1-\gamma_*}$ and $\varepsilon<\frac{(\alpha-1)\pi_{\mathrm{min}}^6}{128\alpha}$. For $0<\alpha<1$, we also have 
    \begin{align*}
        \bar{t}_{\mathrm{mix}}^{(D_\alpha)}(\varepsilon,P)\leq t_{\mathrm{mix}}^{(D_\alpha)}(\varepsilon,P)&\leq \bar{t}_{\mathrm{mix}}^{(D_\alpha)}(\pi_{\mathrm{min}}\varepsilon,P)
        \leq \dfrac{1}{\gamma_*}\left(\ln \dfrac{8\alpha}{(1-\alpha)\pi_{\mathrm{min}}^2}+\ln \dfrac{1}{\varepsilon}\right)\\
        &\leq \bar{t}_{\mathrm{mix}}^{(D_\alpha)}(\varepsilon,P)\dfrac{2\ln \frac{1}{1-\gamma_*}}{\gamma_*}\dfrac{\ln \frac{1}{\varepsilon}+\ln \frac{8\alpha}{(1-\alpha)\pi_{\mathrm{min}}^2}}{\ln \frac{1}{\varepsilon}+\ln \frac{\alpha\pi_{\mathrm{min}}^2}{16}}\\
        &\leq \bar{t}_{\mathrm{mix}}^{(D_\alpha)}(\varepsilon,P)\cdot \dfrac{4}{1-\gamma_*}, \quad 0<\alpha<1,
    \end{align*}
    where in the last inequality we have used $\varepsilon<\frac{\alpha(1-\alpha)\pi_{\mathrm{min}}^6}{2048}$.

    Item \ref{it:averagemix3}: Similar to \eqref{eq:relationship between t D_alpha and t R_alpha}, since the mapping $t\mapsto \frac{1}{\alpha-1}\ln (1+(\alpha-1)t)$ is increasing, we have 
    \begin{equation*}
        d^{(D_{\alpha})}(t)<\varepsilon \iff d^{(R_{\alpha})}(t)<\dfrac{1}{\alpha-1}\left(e^{(\alpha-1)\varepsilon}-1\right),
    \end{equation*}
    which implies 
    \begin{equation*}
        t_{\mathrm{mix}}^{(R_\alpha)}(\varepsilon,P)=t_{\mathrm{mix}}^{(D_\alpha)}\left(\dfrac{1}{\alpha-1}\left(e^{(\alpha-1)\varepsilon}-1\right),P\right).
    \end{equation*}
    Therefore, plugging into item \ref{it:averagemix2} we get the result.
\end{proof}

Theorem \ref{comparison between TV and classical} illustrates that the mixing time under total variation distance and worst case distance are of the same order. As an illustration for Theorem \ref{upper and lower bound of D_alpha} and \ref{comparison between TV and classical}, we consider the following lazy random walk on high-dimensional hypercube.

\begin{example}
[Random walk on the $N$-dimensional hypercube]\label{example:random walk on hypercube}

    Consider $\mathcal{X}=\{-1,1\}^N$ where $N>1$. In this state space, at each step the chain uniformly select a coordinate from $N$ coordinates and update the chosen coordinate to $-1$ and $1$ with probability $\frac{1}{2}$ respectively. For $\bold x, \bold y\in \mathcal{X}$, we denote $\bold x=(x_1,x_2,...,x_N)$ and $\bold y=(y_1,y_2,...,y_N)$, and the transition matrix $P$ of the random walk on $\mathcal{X}$ can be written as
    \begin{equation}\label{eq:transition matrix of random walk}
        P(\bold x,\bold y)=\sum_{i=1}^N\dfrac{1}{N}Q(x_i,y_i)\mathbf{1}_{\{x_j=y_j, j\neq i\}}(\bold x,\bold y),
    \end{equation}
    where $Q$ is the transition matrix defined on $\mathcal{X}_0=
    \{-1,1\}$ such that $Q(x,y)=\frac{1}{2}$ for $x,y\in \mathcal{X}_0$, and therefore the chain under $P$ is lazy. According to \parencite[Example 12.16]{levin2017markov}, there exists $\gamma_*=\frac{1}{N}$ and $\pi_{\mathrm{min}}=2^{-N}$. Plugging into Theorem \ref{upper and lower bound of D_alpha}, if $\varepsilon<\min \left\{\frac{1}{2|\alpha-1|}, \frac{4\alpha}{|\alpha-1|}\right\}$, we have for $\alpha>1$,
    \begin{align*}
        \bar{t}_{\mathrm{mix}}^{(D_\alpha)}(\varepsilon,P)&\leq N\left(N\ln 2+\ln \dfrac{2\alpha}{(\alpha-1)}+\ln \dfrac{1}{\varepsilon}\right), \quad \alpha>1,\\
        \bar{t}_{\mathrm{mix}}^{(D_\alpha)}(\varepsilon,P)&\geq \dfrac{N-1}{2}\left(\ln\dfrac{1}{\varepsilon}-(2N+3)\ln 2\right), \quad \alpha>1,
    \end{align*}
    where we use $\ln \frac{N}{N-1}\leq \frac{1}{N-1}$ in the second inequality. For $\alpha<1$, we similarly have 
    \begin{align*}
        \bar{t}_{\mathrm{mix}}^{(D_\alpha)}(\varepsilon,P)&\leq N\left(N\ln 2+\ln \dfrac{8\alpha}{(1-\alpha)}+\ln \dfrac{1}{\varepsilon}\right), \quad 0<\alpha<1,\\
         \bar{t}_{\mathrm{mix}}^{(D_\alpha)}(\varepsilon,P)&\geq \dfrac{N-1}{2}\left(\ln \alpha+\ln \dfrac{1}{\varepsilon}-(2N+4)\ln 2\right), \quad 0<\alpha<1,
    \end{align*}
    Moreover, plugging $\gamma_*=\frac{1}{N}$ and $\pi_{\mathrm{min}}=2^{-N}$ into Theorem \ref{comparison between TV and classical} and \eqref{eq:upper and lower bound for classical}, when $\varepsilon<2^{-(3N+2)}$, we have
    \begin{equation*}
        \dfrac{(N-1)^2}{2N}\ln \dfrac{1}{2\varepsilon}
        \leq \bar{t}_{\mathrm{mix}}^{(\mathrm{TV})}(\varepsilon,P)\leq N \left(N\ln 2+\ln \dfrac{1}{\varepsilon}\right).
    \end{equation*}
\end{example}

\subsection{Ces\`aro mixing time and ergodicity coefficients}\label{cesaro mixing time}
In Section \ref{mixing time to stationary distribution}, we have investigated various average or worst-case mixing times of Markov chains with irreducible, aperiodic and $\pi$-reversible stochastic matrices. However, the convergence results similar to Theorem \ref{upper and lower bound of D_alpha} and \ref{comparison between TV and classical} may not hold for a periodic transition matrix $P$, as $P$ can have an eigenvalue of $-1$, which implies $\gamma_*=0$. This motivates us to study other notions of mixing times that are suitable for periodic chains. 

One such alternative is to consider the so-called Ces\'aro mixing. Following similar notations in Section \ref{mixing time to stationary distribution}, for irreducible transition matrix $P\in \mathcal{L}$ that is $\pi$-stationary and $\varepsilon>0$, we define several Ces\`aro mixing times that are based upon information divergences as follows:
\begin{align}
    \bar{t}_{\mathrm{Ces}}^{(\mathrm{TV})}(\varepsilon,P)&:=\inf \left\{t>0: \mathrm{TV}\left(\dfrac{1}{t}\sum_{s=1}^t P^s,\Pi\right)<\varepsilon\right\},\label{eq:definition of t_Ces^TV}\\
    \bar{t}_{\mathrm{Ces}}^{(\mathrm{D_\alpha})}(\varepsilon,P)&:=\inf \left\{t>0: D_{\alpha}\left(\dfrac{1}{t}\sum_{s=1}^t P^s\bigg\|\Pi\right)<\varepsilon\right\}, \nonumber\\
    \bar{t}_{\mathrm{Ces}}^{(\mathrm{R_\alpha})}(\varepsilon,P)&:=\inf \left\{t>0: R_{\alpha}\left(\dfrac{1}{t}\sum_{s=1}^t P^s\bigg\|\Pi\right)<\varepsilon\right\}.\nonumber
\end{align}

Next, we offer inequalities between Ces\`aro mixing time and related mixing times as introduced earlier. These for instance can be used for upper bounding the Ces\'aro mixing time. Note that the following results in Theorem \ref{relationship between Cesaro mixing time and classical mixing time} hold without assuming $\pi$-reversibility or aperiodicity of $P$.
\begin{theorem}\label{relationship between Cesaro mixing time and classical mixing time}
    For irreducible $\pi$-stationary $P\in \mathcal{L}$ on $\mathcal{X}$ and $\varepsilon>0$, we have 
    \begin{equation}
        \bar{t}_{\mathrm{Ces}}^{(\mathrm{TV})}(\varepsilon,P)\leq \dfrac{1}{\varepsilon (1-2\varepsilon)}t_{\mathrm{mix}}^{(\mathrm{TV})}(\varepsilon,P),
    \end{equation} 
    where $\bar{t}_{\mathrm{Ces}}^{(\mathrm{TV})}(\varepsilon,P)$  and $t_{\mathrm{mix}}^{(\mathrm{TV})}(\varepsilon,P)$ are introduced respectively in \eqref{eq:definition of t_Ces^TV} and \eqref{t mix worst case}. 
    
    Furthermore, we assume that there exists $0<\theta<1$ such that the ergodicity coefficient $\eta^{(\mathrm{TV})}(P) \leq \theta$ (i.e. $P$ is scrambling). If $\alpha\in (0,1)\cup(1,+\infty)$ and $D_{\alpha}(P\|\Pi)<\infty$, we have 
    \begin{align}
        \bar{t}_{\mathrm{Ces}}^{(\mathrm{TV})}(\varepsilon,P)&\leq \dfrac{1}{\varepsilon}\cdot\bar t_{\mathrm{mix}}^{(\mathrm{TV})}(\varepsilon,P)+\dfrac{1}{1-\theta},\label{eq:ces TV with eta<1}\\
        \bar{t}_{\mathrm{Ces}}^{(D_\alpha)}(\varepsilon,P)&\leq \dfrac{D_{\alpha}(P\|\Pi)}{\varepsilon}\cdot\bar t_{\mathrm{mix}}^{(D_\alpha)}(\varepsilon,P)+\dfrac{1}{1-\theta},\label{eq:ces D_alpha with eta<1}\\
        \bar{t}_{\mathrm{Ces}}^{(R_\alpha)}(\varepsilon,P)&\leq \frac{1+\varepsilon}{\varepsilon}D_{\alpha}(P\|\Pi)\cdot\bar t_{\mathrm{mix}}^{(R_\alpha)}\left(\varepsilon,P\right)+\dfrac{1}{1-\theta},\label{eq:ces R_alpha with eta<1}
    \end{align}
    where we recall that $\bar t_{\mathrm{mix}}^{(\mathrm{TV})}(\varepsilon,P), \bar t_{\mathrm{mix}}^{(D_\alpha)}(\varepsilon,P), \bar t_{\mathrm{mix}}^{(R_\alpha)}(\varepsilon,P)$ are respectively defined in \eqref{t mix TV}, \eqref{t mix D} and \eqref{t mix R}.
\end{theorem}
\begin{proof}
    Denote $t=t_{\mathrm{mix}}^{(\mathrm{TV})}(\varepsilon,P)$, and hence $d^{(\mathrm{TV})}(t)=\max_{x\in\mathcal{X}} \frac{1}{2}\sum_{y\in\mathcal{X}}\left|P^t(x,y)-\pi(y)\right|< \varepsilon$. 
    For any positive integer $l>0$, we have 
    \begin{align*}
        \mathrm{TV}\left(\dfrac{1}{lt}\sum_{s=1}^{lt} P^s,\Pi\right)&=\dfrac{1}{2}\sum_{x,y}\pi(x)\left|\dfrac{1}{lt}\sum_{s=1}^{lt} P^s(x,y)-\pi(y)\right|\\
        &\leq \dfrac{1}{lt}\sum_{x}\pi(x)\sum_{s=1}^{lt}\left\| P^s(x,\cdot)-\pi(\cdot)\right\|_{\mathrm{TV}}\\
        &\leq \dfrac{1}{lt}\sum_{s=1}^{lt}\max_x\left\| P^s(x,\cdot)-\pi(\cdot)\right\|_{\mathrm{TV}}\\
        &= \dfrac{1}{lt}\sum_{s=1}^{lt} d^{(\mathrm{TV})}(s).
    \end{align*}
    According to \parencite[Lemma 4.10,4.11,Ex.4.2]{levin2017markov}, $d(t)$ is non-increasing, and there exists
    \begin{equation*}
        d^{(\mathrm{TV})}(cs)\leq (2d^{(\mathrm{TV})}(s))^c, \quad \text{c is any positive integer},
    \end{equation*}
    hence we have 
    \begin{align*}
        \mathrm{TV}\left(\dfrac{1}{lt}\sum_{s=1}^{lt} P^s,\Pi\right)&\leq \dfrac{1}{lt}\left(t+t\varepsilon+t(2\varepsilon)^2+...+t(2\varepsilon)^{l-1}\right)\\
        &\leq \dfrac{1}{l}\left(\dfrac{1}{1-2\varepsilon}-\varepsilon\right),
    \end{align*}
    and in order for the right hand side above smaller than $\varepsilon$, it suffices to have
    \begin{equation*}
        l\geq \dfrac{1}{\varepsilon(1-2\varepsilon)}\geq \left\lceil \dfrac{1-\varepsilon+2\varepsilon^2}{\varepsilon-2\varepsilon^2}\right\rceil.
    \end{equation*}
    Therefore, 
    \begin{equation*}
        \bar{t}_{\mathrm{Ces}}^{(\mathrm{TV})}(\varepsilon,P)\leq \dfrac{t}{\varepsilon(1-2\varepsilon)}= \dfrac{1}{\varepsilon (1-2\varepsilon)}t_{\mathrm{mix}}^{(\mathrm{TV})}(\varepsilon,P).
    \end{equation*}

    To prove \eqref{eq:ces TV with eta<1} and \eqref{eq:ces D_alpha with eta<1}, by convexity of $D_{\alpha}(\cdot\|\cdot)$ from Proposition \ref{prop:convexity}, for any integer $T>0$, it is easy to verify 
    \begin{align*}
        \mathrm{TV}\left(\dfrac{1}{T}\sum_{s=1}^{T} P^s,\Pi\right)&\leq \dfrac{1}{T}\sum_{s=1}^{T} \mathrm{TV}(P^s,\Pi),\\
        D_{\alpha}\left(\dfrac{1}{T}\sum_{s=1}^{T} P^s \bigg \|\Pi\right)&\leq \dfrac{1}{T}\sum_{s=1}^{T} D_{\alpha}(P^s\|\Pi).
    \end{align*}
    Denote $t=\bar t_{\mathrm{mix}}^{(\mathrm{TV})}(\varepsilon,P)$, which gives $\mathrm{TV}(P^t\|\Pi)\leq \varepsilon$. According to Proposition \ref{prop:ergocoef} item \ref{it:ergcoef1}, we thus have, for $s>t$,
    \begin{align*}
        \mathrm{TV}(P^s,\Pi)\leq \theta^{s-t}\varepsilon, 
    \end{align*}
    which yields 
    \begin{align*}
        \mathrm{TV}\left(\dfrac{1}{T}\sum_{s=1}^{T} P^s,\Pi\right)&\leq \dfrac{1}{T}\left(\sum_{s=1}^{t}\mathrm{TV}(P^s,\Pi)+\sum_{s=t+1}^T \mathrm{TV}(P^s,\Pi)\right)\\
        &\leq \frac{1}{T}\left(t+\frac{\varepsilon \theta}{1-\theta}\right).
    \end{align*}
    We set the right hand side above to be smaller than $\varepsilon$. As such we require that
    \begin{equation*}
        T\geq \frac{t}{\varepsilon}+\frac{1}{1-\theta}\geq\left\lceil \frac{t}{\varepsilon}+\frac{\theta}{1-\theta}\right\rceil,
    \end{equation*}
    and therefore \eqref{eq:ces TV with eta<1} is obtained. The proof for \eqref{eq:ces D_alpha with eta<1} is similar by recalling that $\eta^{(D_\alpha)}(P) \leq \eta^{(\mathrm{TV})}(P)\leq \theta$ as stated in Proposition \ref{prop:ergocoef} item \ref{it:ergcoef1}. For \eqref{eq:ces R_alpha with eta<1}, analogous to the proof of \eqref{eq:relationship between t D_alpha and t R_alpha} we see that
    \begin{equation*}
        \bar{t}_{\mathrm{Ces}}^{(R_\alpha)}(\varepsilon,P)=\bar{t}_{\mathrm{Ces}}^{(D_\alpha)}\left(\dfrac{1}{\alpha-1}\left(e^{(\alpha-1)\varepsilon}-1\right),P\right),
    \end{equation*}
    then by \eqref{eq:ces D_alpha with eta<1} and \eqref{eq:relationship between t D_alpha and t R_alpha} we have
    \begin{align*}
        \bar{t}_{\mathrm{Ces}}^{(R_\alpha)}(\varepsilon,P)&=\bar{t}_{\mathrm{Ces}}^{(D_\alpha)}\left(\dfrac{1}{\alpha-1}\left(e^{(\alpha-1)\varepsilon}-1\right),P\right)\\
        &\leq \dfrac{\alpha-1}{e^{(\alpha-1)\varepsilon}-1}D_{\alpha}(P\|\Pi)\bar t_{\mathrm{mix}}^{(D_\alpha)}\left(\dfrac{1}{\alpha-1}\left(e^{(\alpha-1)\varepsilon}-1\right),P\right)+\dfrac{1}{1-\theta}\\
        &=\dfrac{\alpha-1}{e^{(\alpha-1)\varepsilon}-1}D_{\alpha}(P\|\Pi)\bar t_{\mathrm{mix}}^{(R_\alpha)}\left(\varepsilon,P\right)+\dfrac{1}{1-\theta}\\
        &\leq \frac{1+\varepsilon}{\varepsilon}D_{\alpha}(P\|\Pi)\bar t_{\mathrm{mix}}^{(R_\alpha)}\left(\varepsilon,P\right)+\dfrac{1}{1-\theta},
    \end{align*}
    where the last inequality comes from $e^{(\alpha-1)\varepsilon}\geq (\alpha-1)\varepsilon+1$ for $\alpha>1$ and $e^{(1-\alpha)\varepsilon}\geq (1-\alpha)\varepsilon +1$ for $0<\alpha<1$. 
\end{proof}

In \eqref{eq:ces TV with eta<1}, \eqref{eq:ces D_alpha with eta<1} and \eqref{eq:ces R_alpha with eta<1}, we require that $P$ is a scrambling matrix so that $\eta^{(\mathrm{TV})}(P)<1$ (see Proposition \ref{prop:ergocoef} item \ref{it:ergcoef7} and Remark \ref{remark: scrambling}), while in many cases this condition is hard to guarantee. An alternative consideration is, if we can show that $\eta^{(\mathrm{TV})}(P^t)<1$ for some $t>0$, then by  submultiplicativity of $\eta^{(\mathrm{TV})}(\cdot)$ from Proposition \ref{prop:ergocoef} item \ref{it:submulti}, it is a direct corollary that for all $s\geq t$, $\eta^{(\mathrm{TV})}(P^s)\leq \eta^{(\mathrm{TV})}(P^t)<1$, and we can still obtain similar bound for Ces\`aro mixing time. In the following example, we will show that in general Metropolis-Hastings algorithms under low temperature, the first time that the ergodicty/Dobrushin coefficient $\eta^{(\mathrm{TV})}(P^t)<\varepsilon$ for any $\varepsilon>0$ has an exponential dependence on the exit height.
\begin{example}
    [Ergodicity/Dobrushin coefficient of Metropolis-Hastings chain]

    Let $Q$ be a reversible transition matrix with respect to probability measure $\mu$ on finite state space $\mathcal{X}$. Let $U:\mathcal{X}\rightarrow \mathbb R$ be the target function, $\beta$ be the inverse temperature, and we use the Metropolis-Hastings algorithm to sample from the Gibbs distribution $\pi_\beta(x) \propto \mu(x)e^{-\beta U(x)}$. The transition matrix of the chain can be written as
    \begin{equation*}
        P_{\beta}(x,y)=\begin{cases}
            Q(x,y)\exp \left(-\beta (U(y)-U(x))_+\right), \quad & x\neq y,\\
            1-\sum_{z\neq x}P_{\beta}(x,z), \quad & x=y.
        \end{cases}
    \end{equation*}
    Suppose the chain generated by $P_\beta$ is $\{X_n\}_{n\geq 0}$. Following \parencite{catoni2006simulated} and \parencite{miclo2002relaxation}, for $A\subseteq \mathcal{X}$, we define the first hitting time of the set $A$ to be $\tau_A:=\inf\{n: X_n\in A\}$, and the exit height $H(A)$ as
    \begin{align*}
        h(x,A)&:=\lim_{\beta\rightarrow \infty} \dfrac{1}{\beta} \ln \left(\mathbb E_{x} [\tau_{A^c}]\right),\\
        H(A)&:=\max_{x\in A} h(x,A).
    \end{align*}
    Note that there is a dependency on the inverse temperature $\beta$ of $\mathbb E_{x} [\tau_{A^c}]$ via the underlying Metropolis-Hastings chain. Next, we define the product chain
    \begin{equation*}
        P_{\beta}^{\otimes 2}\left((x_1,x_2),(y_1,y_2)\right):=P_\beta (x_1,y_1)P_\beta (x_2,y_2), \quad (x_1,x_2),(y_1,y_2)\in \mathcal{X}^2,
    \end{equation*}
    and the constant 
    \begin{equation*}
        H_3:=H(\mathcal{X}^2\setminus\Delta), \quad \Delta:=\{(x,x)\in \mathcal{X}^2:x\in \mathcal{X}\},
    \end{equation*}
    where $H_3$ is an exit height with respect to the product chain $P_{\beta}^{\otimes 2}$.

    For any $0<\varepsilon<1$, we define the first time that $\eta^{(\mathrm{TV})}(P_\beta^t)<\varepsilon$ as 
    \begin{align*}
        t_{\mathrm{Dob}}^{(\mathrm{TV})}(\varepsilon, P_{\beta})&:=\inf\{t\geq 0: \eta^{(\mathrm{TV})}(P_{\beta}^t)<\varepsilon\},\\
        \widetilde t_{\mathrm{Dob}}^{(\mathrm{TV})}(\varepsilon, P_{\beta})&:=\inf\{t\geq 0: \widetilde \eta^{(\mathrm{TV})}(P_{\beta}^t)<\varepsilon\},
    \end{align*}
    then according to \parencite[Proposition 2.1]{miclo2002relaxation}, there exists
    \begin{equation*}
        \lim_{\beta\rightarrow \infty}\dfrac{1}{\beta} \ln \left( \widetilde t_{\mathrm{Dob}}^{(\mathrm{TV})}(\varepsilon, P_{\beta})\right)=H_3.
    \end{equation*}
    Since $\eta^{(\mathrm{TV})}(P_{\beta}^t)=\widetilde \eta^{(\mathrm{TV})}(P_{\beta}^t)$ by Proposition \ref{prop:ergocoef} item \ref{it:ergcoef2}, we then have 
    \begin{equation*}
        \lim_{\beta\rightarrow \infty}\dfrac{1}{\beta} \ln \left( t_{\mathrm{Dob}}^{(\mathrm{TV})}(\varepsilon, P_{\beta})\right)=H_3.
    \end{equation*}
\end{example}

\section*{Acknowledgements}
Youjia Wang gratefully acknowledges the financial support from National University of Singapore via the Presidential Graduate Fellowship. Michael Choi acknowledges the financial support of the project “MAPLE: Mechanistic Accelerated Prediction of Protein Secondary Structure via LangEvin Monte Carlo” with grant number 22-5715-P0001 under the National University of Singapore Faculty of Science Ministry of Education Tier 1 grant Data for Science and Science for Data collaborative scheme, as well as the startup funding of the National University of Singapore.

\printbibliography


\end{document}